%% file: trial.tex
\documentclass[11pt]{article}

\usepackage{fullpage,amsmath,xspace,amssymb,epsfig,amsthm}
\usepackage{algorithm,caption}
\usepackage{algorithmic,fullpage}
\usepackage{amsmath,xspace,amssymb,epsfig}
\usepackage{algorithm}
\usepackage{algorithmic}
\usepackage{paralist}

\newtheorem{Thm}{Theorem}
\newtheorem{Lem}[Thm]{Lemma}
\newtheorem{Cor}[Thm]{Corollary}
\newtheorem{Prop}[Thm]{Proposition}
\newtheorem{Claim}{Claim}

\newcommand\mbZ{\mbox{$\mathbb{Z}$}}
\newcommand\mbR{\mbox{$\mathbb{R}$}}

\newcommand\Bn{\{0,1\}^n}

\newcommand {\ie} {\textit{i.e.}\xspace}
\newcommand {\st} {\textit{s.t.}\xspace}

\newcommand\valpha{\mbox{$\boldsymbol\alpha$}}

\newcommand\co{\mbox{\bf co-}}
\newcommand\am{\mbox{\bf AM}\xspace}

\newcommand\p{\mbox{\bf P}\xspace}
\newcommand\sharpp{\mbox{{\small \#}{\bf P}}\xspace}
\newcommand\np{\mbox{\bf NP}\xspace}
\newcommand\ph{\mbox{\bf PH}\xspace}
\newcommand\pp{\mbox{\bf PP}\xspace}
\newcommand\fp{\mbox{\bf FP}\xspace}

\newcommand\ppad{\mbox{\bf PPAD}\xspace}

\newcommand{\uu}{\textup{${\sf _{u}}$}\xspace} 
\newcommand{\V}{\mbox{\sf {V}}\xspace} 
\newcommand{\A}{\mbox{\sf {A}}\xspace} 

\newcommand\GraphI{\mbox{\sf {GraphIso}}\xspace}

\newcommand\Clique{\mbox{\sf {Clique}}\xspace}
\newcommand\HC{\mbox{\sf {HamiltonianCycle}\xspace}}

\newcommand{\GpI}{\mbox{\sf {GroupIso}}\xspace}
\newcommand{\GpIp}{\mbox{\sf {GroupIso}$(\cdot, \mbZ_p)$}}

\newcommand\PHC{\mbox{\sf {PromisedHamiltonianCycle}\xspace}}
\newcommand\yes{\mbox{\sf {Yes}}\xspace}
\newcommand\no{\mbox{\sf {No}}\xspace}

\newcommand{\Sat}{\mbox{\sf {SAT}}\xspace}

\newcommand\ssum{\mbox{\sf {SubsetSum}}\xspace}
\newcommand\nash{\mbox{\sf {Nash}}\xspace}
\newcommand\Core{\mbox{\sf {Core}}\xspace}

\newcommand\SM{\mbox{\sf {StableMatching}}\xspace}
\newcommand\Sort{\mbox{\sf {Sorting}}\xspace}

\newcommand\mcA{\mathcal{A}\xspace}
\newcommand\mcB{\mathcal{B}\xspace}
\newcommand\mcO{\mathcal{O}\xspace}
\newcommand\D{\mathsf{D}}
\newcommand\R{\mathsf{R}}

\begin{document}
\title{\bf{On the Complexity of Trial and Error}}

\author{
Xiaohui Bei\thanks{Nanyang Technological University, Singapore. Email: {\tt xhbei@ntu.edu.sg}.}
\and Ning Chen\thanks{Nanyang Technological University, Singapore. Email: {\tt ningc@ntu.edu.sg}.}
\and Shengyu Zhang\thanks{The Chinese University of Hong Kong, Hong Kong. Email: {\tt syzhang@cse.cuhk.edu.hk}.}
}

\date{
\vspace{0.2in}
\begin{quote}
\begin{flushright}
{\it Great scientific discoveries have been made by men} \\
{\it seeking to verify quite erroneous theories about the nature of things.} \\[.1in]
 -- {\it Aldous Huxley}
\end{flushright}
\end{quote}
}
\maketitle

\begin{abstract}
Motivated by certain applications from physics, biochemistry, economics, and computer science in which the objects under investigation are unknown or not directly accessible because of various limitations, we propose a trial-and-error model to examine search problems in which inputs are {\em unknown}.
More specifically, we consider constraint satisfaction problems $\bigwedge_i C_i$, where the constraints $C_i$ are hidden, and the goal is to find a solution satisfying all constraints. We can adaptively propose a candidate solution (i.e., {\em trial}), and there is a verification oracle that either confirms that it is a valid solution, or returns the index $i$ of a violated constraint (i.e., {\em error}), with the exact content of $C_i$ still hidden.

We studied the time and trial complexities of a number of natural CSPs, summarized as follows. On one hand, despite the seemingly very little information provided by the oracle, efficient algorithms do exist for Nash, Core, Stable Matching, and SAT problems, whose unknown-input versions are shown to be as hard as the corresponding known-input versions up to a factor of polynomial. The techniques employed vary considerably, including, e.g., order theory and the ellipsoid method with a strong separation oracle.

On the other hand, there are problems whose complexities are substantially increased in the unknown-input model. In particular, no time-efficient algorithms exist for Graph Isomorphism and Group Isomorphism (unless \ph collapses or $\p = \np$). The proofs use quite nonstandard reductions, in which an efficient simulator is carefully designed to simulate a desirable but computationally unaffordable oracle.

Our model investigates the value of input information, and our results demonstrate that the lack of input information can introduce various levels of extra difficulty. The model accommodates a wide range of combinatorial and algebraic structures, and exhibits intimate connections with (and hopefully can also serve as a useful supplement to) certain existing learning and complexity theories.
\end{abstract}

\setcounter{page}{0}
\thispagestyle{empty}

\newpage

\section{Introduction}

\input{Motivation.tex}

\input{Model.tex}

\input{ours.tex}

\input{Related-Work.tex}

\input{Stable-Matching.tex}

\input{SAT}

\input{Group-ISO.tex}

\input{Graph-ISO.tex}

\input{nash.tex}

\input{core.tex}

\input{subset.tex}

\section{Concluding Remarks}\label{sec:conclusion}

In this paper, we propose a trial and error model to investigate search problems with unknown input. We consider a number of natural problems, and show that the lack of input knowledge may introduce different levels of extra difficulty in finding a valid solution. Our complexity results range from polynomially solvable, to \np-complete and exponential. 
Our model and results demonstrate the value of input information in solution finding from the computational complexity viewpoint.

The present work showcases a number of algorithms and lower bounds. Meanwhile, a number of important questions are left for future exploration. Closing the small gaps in Theorem \ref{thm:trialbounds} and examining more CSP problems are the obvious and specific ones. The following is a list of more problems and directions for further research.
\begin{itemize}
\item Information processing on hidden inputs is a common phenomenon in many scenarios, and the present work tries to address the related computational complexity issues in a specific and natural framework. What other general frameworks could be employed for systematic studies of hidden inputs from an algorithmic perspective?


\item Our complexity results focus on either the trial or the time cost. It would be intriguing to consider the tradeoff between them. For instance, in $\Sort\uu$ and $\SM\uu$, our deterministic upper bounds $O(n\log n)$ and $O(n^2\log n)$ for trial complexity are established by exponential-time algorithms. If we only allow {\em polynomial time} computation, then we do not have any bound better than $O(n^2)$ and $O(n^3)$ for $\Sort\uu$ and $\SM\uu$, respectively. (Note that the classic argument of graph entropy for sorting under a partial order~\cite{KK95} is not directly applicable here, as the allowed queries there are of the standard form of pair comparison.)

    On the lower bound side, a natural question is whether any lower bound better than $\Omega(n\log n)$ and $\Omega(n^2\log n)$ can be proven for $\Sort\uu$ and $\SM\uu$ with polynomial time computation power. Note that the bound, if possible, would probably be very difficult to prove because it implies that $\sharpp \neq \fp$ and thus $\p \neq \pp$. (If $\sharpp = \fp$, then counting linear extensions, as a $\sharpp$-complete problem, can be computed in polynomial-time, which makes our algorithms for $\Sort\uu$ and $\SM\uu$ also in polynomial time.)

\item It is well known in decision tree complexity that deterministic and randomized complexities can be polynomially separated~\cite{BdW02}, and a fundamental open question is whether the gap exhibited by the NAND tree~\cite{SW86} is the largest possible. What separation between deterministic and randomized trial complexities could we have in our model? This question could also be considered in the polynomial-time computation framework.


\item An algorithm in our model can access two oracles, verification and computation. In this paper, we consider only the complexity that interacts with the verification oracle. It is natural to ask about the query complexity of the other oracle (the problem is of particular importance when the computational complexity of the problem itself is large).
\end{itemize}

\subsection*{Acknowledgments}
We are grateful to Shang-Hua Teng, Leslie Valiant, Umesh Vazirani, and Andrew Yao for their helpful discussions and comments. We also thank Eric Allender, Graham Brightwell, Leslie Goldberg, and Kazuo Iwama for pointing out~\cite{Sch87}, \cite{BW91}, \cite{BD98}, and~\cite{Ng89}, respectively.

\small
\bibliographystyle{abbrv}
\bibliography{reference}
\normalsize

\end{document}

%% file: Motivation.tex
In a broad sense, computer science studies computation and other information processing tasks. Theoretical computer science, in particular, focuses on understanding the ultimate power and limits of computation in various models. A central question in theoretical computer science is to find the minimum cost of an algorithm for computing a function $f$ on inputs $x$. Algorithm design and computational complexity analysis assume that the input is given explicitly.
However, in many scenarios, we actually \emph{lack input information}.
\begin{itemize}
\item In a normal-form game, it is usually assumed that the payoff function of every player is given explicitly (or can be computed easily in a certain way). However, in many circumstances, players do not necessarily know their payoffs or possible strategies, particularly when they are exploring a new environment (such as a new business model). In some cases, they may not even know the number of other players, let alone their strategies~\cite{Hal08}.

\item In a two-sided matching marketplace, every individual has a preference over the agents of the other side. 
However, the individuals themselves may not know their (precise) preferences.
    For example, in a job market, an applicant may not know precisely how much he or she would like the job in question because of a lack of information about the nature of the job, the company culture, and his or her future relations with colleagues. At the same time, it is generally quite difficult for recruiters to judge which candidates best fit the position. 
    Indeed, decision makers quite often make hiring mistakes: ``...~a systematic and continuous approach to fitting the right person to the right job at the right time has long been the Holy Grail of workforce organization''~\cite{Mck03}.




\item 

    In the event of an infectious disease outbreak, due to an unknown virus, biochemists need to find diagnostic reagents that have no serious side effects. In a simplified formulation, this involves a search for a reagent that satisfies a collection of constraints (e.g., one constraint may be that the reagent should not contain certain medical ingredients composed in a certain way, as otherwise, its reaction with the virus could cause a severe headache). If the biochemists knew everything about the virus (e.g., its DNA sequence, chemical composition, etc.),
    then it would be much easier for them to find a diagnostic reagent. If the virus is largely unknown, however, then they are left with an effectively unknown-constraints formula to satisfy. Of course, they could try to employ modern DNA technologies to gain most information on the virus, but doing so usually takes a long time, and identifying a diagnostic reagent to control the ensuing pandemic is a matter of urgency.
\end{itemize}

In summary, an input, while it exists, may be {\em unknown} because of our limited knowledge and control of the system, or of our lack of experience in a new environment.
There are numerous other scenarios with unknown inputs, e.g., animal behavior studies, neural science, and hidden web databases,~to~name~just~a~few.

\subsection{Trial and Error}

Trial and error is a basic methodology in problem solving and knowledge acquisition, and it has also been used extensively in product design and experiments~\cite{Mon08}. Generally speaking, the approach proceeds by adaptively posing a sequence of candidate solutions and observing their validity. If a proposed candidate solution is found to be valid, then the mission is accomplished. Otherwise, an error is signaled from one of the characteristics of the studied object. An important feature of the approach is its solution orientation: the goal is to find one solution, with little care paid to other considerations such as why the solution works~\cite{wiki-trial}.

Trial and error is also a commonly used approach in the aforementioned examples. In economics, individuals adopt and adaptively adjust their strategies based on observed market reactions. Such self-motivated, but self-regulating, types of behavior, as implied by Adam Smith's ``invisible hand'' theory, can converge to a socially desirable state (even without the individuals involved having any knowledge of one another).
In a company, an employee will usually look for a more suitable position when dissatisfied with his or her current one, and senior management will usually encourage personnel adjustments to enhance performance.
Biomedical scientists conduct clinical trials to test their designed reagents, and if an unacceptable side effect is observed, then they collect and analyze feedback data to help with future diagnostic reagent tests~\cite{Indrayan04}.

The most critical ingredient in the trial and error approach is how to employ previously returned errors to propose future trials. This procedure is algorithmic in nature, but it does not seem to have been formally addressed from an algorithmic perspective. This paper aims to investigate this approach on a broad category of problems with unknown inputs.

%% file: Model.tex
\subsection{Model and Preliminaries}


Motivated by the foregoing examples, we investigate the effects of the lack of input information from a computational viewpoint on the basis of the trial and error approach.
The central question we consider is the following.

\begin{quote}
{\em How much extra difficulty is introduced due to the lack of input knowledge?}\\[-.2in]
\end{quote}

In this paper, we explore this question in search problems. 
Suppose that on an input (instance) $I$, there is a set $S(I)$ of \emph{solutions}. A search problem is to find a solution $s\in S(I)$ to input $I$.\footnote{One might wish to find more or even all solutions. Here, we follow the standard requirement for searching problems in complexity theory~\cite{Gol08,AB09} by asking for only one (arbitrary) solution.}
Numerous problems arising from a variety of applications studied in algorithm design and computational complexity are search problems. Typical examples include searching for a Nash equilibrium in a multi-player game, searching for a satisfying assignment in a conjunctive normal form (CNF) formula, and finding a stable matching to pair individuals with preferences in a two-sided market.

All of these problems, in addition to the motivating examples discussed earlier, naturally fall into the broad category of \emph{constraint satisfaction problems} (CSPs). Suppose that there is a space $\Omega = \Bn$ of candidate solutions. Corresponding to an input $I$, there are a number of constraints $C_1, C_2, \ldots, C_m (, \ldots)$, where each $C_i\subseteq \Bn$ is a relation on the solution variables defined on given domains.\footnote{Note that it is possible for a problem to have different CSP definitions, which, depending on the available set of tools, may in turn lead to different complexities for solving the problem.
For example, to identify an unknown substance, we can employ different (physical, chemical, etc.) test methods, which, in general, require different costs.
This phenomenon reflects the intrinsic variety of a problem and the diversity of its solutions.
Thus, to define an unknown-input problem, we need to indicate explicitly its corresponding CSPs. For the problems investigated in this paper, we arguably use their most natural definitions.}
The solutions of $I$ are defined as those $s$ that satisfy all constraints $C_i$, i.e., $s\in \bigcap_i C_i$.
Note that the number of constraints can range from constant to polynomial, exponential, or even infinite.
CSPs are a subject of intensive research in theoretical computer science, artificial intelligence, and operations research, and they provide a common basis for exploration of a large number of problems with both theoretical and practical importance.

This paper addresses the situation in which the input $I$ is unknown. For a search problem $\A$, we denote by $\A\uu$ the same search problem with unknown inputs.
For example, in the \SM problem, the input contains the preference lists of all men and women; in $\SM\uu$, these preference lists are unknown to us. The constraints are that all man-woman pairs $(m,w)$ are not blocking pairs, and the task is to find a solution that satisfies all constraints, namely a stable matching~\cite{GS62}.

Similar to the way in which a biochemist
proposes a chemical reagent and then performs clinical tests,
here our method of searching for a solution of a CSP is also the trial and error approach. We propose a candidate solution $s$: If $s$ is not a valid solution, then we are told so by a \emph{verification oracle} $\V$, and, what is more, $\V$ also gives us the \emph{index} of one constraint that is not satisfied.
Otherwise, $s$ satisfies all constraints, and then we cannot observe any violated constraint; equivalently, $\V$ returns a confirmative answer, and our job is done. Some remarks are necessary: 
\begin{itemize}
\item If more than one constraint is violated, then (the index of) any one of them can be returned by $\V$. We make no assumption about which one, not only because worst-case analysis is standard in algorithm and complexity studies, but also because in many applications, such as drug tests, the verification oracle is carried out by Nature or human bodies, and thus how and which violation is returned is truly beyond our current understanding.
\item Note that $\V$ does not reveal the constraint itself, but only its index or label.
    For example, we know something like ``the third constraint is violated'' in the proposed assignment of the $\Sat\uu$ problem, or ``the second player has a better mixed strategy'' for the proposed strategy in the $\nash\uu$ problem,
	but the exact content of the constraint (i.e., the literals in the clause of $\Sat\uu$ or the player's utility function of $\nash\uu$) is still unknown to us,
    which is consistent with our motivating examples. If a headache is observed in a drug development clinical trial, then we do not always know which components of the proposed reagent caused the problem: We have only a label of ``headache'' for the proposed reagent.
\end{itemize}
Surprisingly, despite this seemingly very little information and the worst-case assumption on the verification oracle, we still have efficient algorithms for many problems.

Given the verification oracle \V, an algorithm is an interactive process with \V. We choose candidate solutions (i.e., trials), and the oracle returns violations (i.e., errors). The process is adaptive, i.e., the newly proposed solution can be based on the historical information returned by the oracle. 

Because our focus is on how much \emph{extra} difficulty is introduced by the lack of input information for a search problem $\A$, we single out this complexity by comparing the unknown-input and known-input scenarios. To this end, we equip our algorithms with another oracle, the \emph{computation oracle}, which can solve the known-input version of the same problem $\A$.
Overall, our algorithms can access two oracles, the verification oracle and the computation oracle (we do not allow them to invoke each other).

As is standard in complexity theory, a query to either oracle has a unit time cost. The \emph{time complexity} of a problem with unknown inputs is the minimum time needed for an algorithm to solve it for all inputs and all verification oracles consistent with the input. We employ the standard notation in computational complexity theory for complexity classes such as $\p$ and $\np$ and also for oracles. For example, $\A\uu \in \p^{\scriptsize \V,\A}$ means that problem $\A\uu$ can be solved by a polynomial-time algorithm with verification oracle $\V$ and the computation oracle that can solve the known-input version of $\A$. If this occurs, then we consider the extra complexity (resulting from the unknown input) not to be very high.
The central question can therefore be translated to the following. Given a search problem $\A$, is $\A\uu\in \p^{{\scriptsize \V,\A}}$?
If the given known-input problem $\A$ is in $\p$, then the computation oracle can be omitted, and the problem becomes ``Is $\A\uu\in \p^{{\scriptsize \V}}$?''

We also define the \emph{trial complexity} of an unknown-input problem $\A\uu$ as the minimum number of queries to the verification oracle that any algorithm needs to make, regardless of its computational power.\footnote{It is thus the ``query complexity'' to the verification oracle. Here, we adopt the term ``trial complexity'' to avoid any potential confusion of the two types of oracle queries (corresponding to the two oracles).} As is standard in query complexity theory, we can consider deterministic or (Las Vegas) randomized algorithms. The latter can be assumed to be error-free because of the verification oracle $\V$, and we count the cost as the expected number of queries to $\V$. We denote by $\D(A\uu)$ and $\R(A\uu)$ the deterministic and randomized trial complexities of $\A\uu$, respectively.

We investigate trial complexity not only because it provides a rigorous proof of computational hardness, but also because it measures the number of trials (or in another perspective, errors) that must be undertaken to find a solution. Note that in many scenarios, such as diagnostic reagent development, trials constitute the major expense, both financially and temporally, and in almost all of the motivating examples discussed earlier, an important goal is to design protocols or experiments with a small number of trials.

%% file: ours.tex
\subsection{Our Results and Techniques}

We consider a number of problems, that are motivated by the aforementioned examples, to investigate the trial and time complexities resulting from the lack of input knowledge. (The formal definitions of these problems and their natural formulation as CSPs are deferred to subsequent sections.)

\begin{Thm}
  For the following problems $\A$, \mbox{we have $\A\uu \in \p^{\scriptsize \V, \A}$}.
  \begin{itemize}
    \item \nash: Find a Nash equilibrium of a normal-form game.
    \item \Core: Find a core of a cooperative game.
    \item \SM: Find a stable matching of a two-sided market with preferences.
    \item \Sat: Find a satisfying assignment of a CNF formula.
  \end{itemize}
\end{Thm}

\nash is a fundamental problem in game theory, and its complexity has been characterized (as $\ppad$-complete)~\cite{DGP10,CDT09}. \Core is also a fundamental problem in cooperative game theory~\cite{NRT07}. Both problems are naturally defined as CSPs.
Our algorithms for both $\nash\uu$ and $\Core\uu$ employ the ellipsoid method, although for $\nash\uu$ we shrink the input space, and for $\Core\uu$ we shrink the solution space. One technical difficulty is that the target space may degenerate to the case of containing at most one point. (In $\nash\uu$, there is only one input point, and in $\Core\uu$ the core may contain only one point or even be empty.) Note that the standard perturbation approach, which proceeds by increasing the volume of the feasible region, is not applicable in our setting, because the linear constraints, as the input, are unknown. Here, we employ a more sophisticated ellipsoid method that works as long as the polyhedron can be specified by a \emph{strong separation oracle}. As it turns out, this oracle can be constructed from the verification oracle $\V$ in both problems, and, crucially, the construction for $\nash\uu$ uses the existence of a~Nash~equilibrium~in~\emph{any}~game.

$\SM$ is a problem with interesting combinatorial structures and many applications, such as the pairing of graduating medical students with hospital residencies~\cite{RS92,Roth08}.
$\Sat$ is a natural CSP, with the constraints being the OR of some literals.
Considering the practical significance of $\SM$ and $\Sat$, we take a closer look at their trial complexities.

\begin{Thm}\label{thm:trialbounds}
	We have the following bounds for the trial complexity.
    \begin{itemize}
      \item $\Omega(n^2) \leq \R(\SM\uu) \le \D(\SM\uu) \leq O(n^2\log n)$, where $n$ is the number of agents.
      \item Given a formula with $n$ variables and $m$ clauses, $\R(\Sat\uu) \\ \le \D(\Sat\uu) = O(mn)$. Further, $\R(\Sat\uu) = \Omega(mn)$ if $m = \Omega(n^2)$, and $\R(\Sat\uu) = \Omega(m^{3/2})$ if $m = o(n^2)$.
    \end{itemize}
\end{Thm}

The proofs of both lower bounds deviate from the standard method of applying Yao's min-max principle. Rather, they are obtained by arguing that, for an arbitrary but fixed \emph{randomized} algorithm with an insufficient number of queries, there are input instances with disjoint solution sets between which the algorithm cannot distinguish. The existence of such input instances is proved by the probabilistic method for $\SM\uu$, and by an adaptive construction procedure for $\Sat\uu$.


The upper and lower bounds proofs for $\R(\SM\uu)$ also employ order theory~\cite{Bri99,DP02}. A key step is to characterize how fast one can shrink the set of linear orders consistent with a partial order by worst-case pair violations. We identify the average height as the correct measure; the control of which allows us to bound the speed of the shrinkage. 
Along the way, we examine another natural problem, $\Sort\uu$, whose trial complexity is completely pinned down as $\Theta(n\log n)$.

\medskip
It is somewhat surprising that knowing only the indices of violated constraints 
is already sufficient to admit quite a number of efficient algorithms. 
It is therefore natural to wonder whether 
the lack of input information adds any extra difficulty at all in finding a solution.
We find that it does indeed: there are problems whose unknown-input versions are considerably more difficult than their known versions. Two representatives are $\GraphI$ and $\GpI$, the problems of deciding whether two given graphs or groups are isomorphic.

\begin{Thm}
We have the following hardness results.
    \begin{itemize}
      \item If $\GraphI\uu\in \p^{\scriptsize \V, \GraphI}$, then
      the polynomial hierarchy ($\ph$) collapses to the second level.
      \item If $\GpI(\cdot,\mbZ_p)\uu\in \p^{\scriptsize \V}$, then we have $\p=\np$. Here, $\GpI(\cdot,\mbZ_p)$ is the group isomorphism problem with the second group known as $\mbZ_p$ for a prime $p$.
    \end{itemize}
However, if $\Sat$ is given as the computation oracle, then we have deterministic polynomial-time algorithms for $\GraphI$ and $\GpI$, i.e., $\GraphI\uu \in \p^{\scriptsize \V, \Sat}$ and $\GpI\uu\in \p^{\scriptsize \V, \Sat}$, with $O(n^6)$ and $O(n^2)$ trials, respectively. 
\end{Thm}

Note that $\GpI(\cdot,\mbZ_p)$ (with a known input) admits a simple polynomial-time algorithm by comparing the multiplication tables. 
Actually, $\GpI$ is in $\p$ if the two groups are Abelian~\cite{Kav07}. However, if the multiplication table of the input group is unknown, then, surprisingly, the problem becomes \np-hard. 
Interestingly, this substantial increase in computational difficulty occurs only for time complexity, not for trial complexity, which can be seen as a tradeoff (from below) between the two complexity measures---a phenomenon not commonly seen in other query models.

This hardness result for $\GpI(\cdot,\mbZ_p)\uu$ is proved by a nonstandard reduction from the classic \np-complete problem of finding a Hamiltonian cycle. We use an algorithm $\mcA$ for $\GpI(\cdot,\mbZ_p)\uu$ to find a Hamiltonian cycle in a given graph $G$ in the following way. Assuming the existence of a Hamiltonian cycle $C$, which does not change the $\np$-completeness of the problem, we define a group $H$ via $C$ and run $\mcA$ on input $(H,\mbZ_p)$. An issue here is that because the reduction algorithm has only polynomial time, it cannot find such a Hamiltonian cycle for defining $H$. A related issue is how to provide the verification oracle $\V$ for $\mcA$ without knowing $C$. These issues can be overcome by (i) making use of the crucial property that $\mcA$ does not know its input, and (ii) designing an efficient simulator $\V'$ for the verification oracle $\V$. Due to the time constraint, $\V'$ cannot perfectly mimic $\V$ to answer all of $\mcA$'s questions correctly. However, it is designed with the favorable property that whenever it produces an incorrect answer, a Hamiltonian cycle in $G$ has just been found. ($\mcA$'s correctness on $H$ may already have been compromised, but we are not further concerned with it. We simply use $\mcA$'s code to serve the purpose of finding a Hamiltonian cycle in $G$.)


\smallskip
Finally, beyond all of the foregoing problems that can be solved in $\p^{\scriptsize \V, \Sat}$, we show via an information theoretical argument that certain other problems, such as Subset Sum, have an exponential lower bound for the randomized trial complexity.

\begin{Thm}
  $\R(\ssum\uu) = \Omega(2^n)$.
\end{Thm}

In a followup work~\cite{BCZ12}, we showed that solving an unknown linear programming with $m$ constraints and $n$ variable requires $\Omega(m^{n/2})$ queries to the oracle, i.e., it also has an exponential lower bound on the trial complexity. This result implies that our efficient algorithm solving $\nash\uu$ does not completely resort to the computation oracle $\nash$: it is indeed the combinatorial structure of Nash equilibrium (i.e., its existence) that yields our algorithm. See more discussions in~\cite{BCZ12}.

Our results illustrate the variety of time and trial complexities that arise from the lack of input information for different problems, and imply distinct levels of the cruciality of input information for different problems.



\smallskip
In addition to the specific techniques previously mentioned for each problem, a general remark is that, at a very high level, our algorithms are in line with the candidate elimination approach, similar to many existing learning algorithms~\cite{KV94}. However, our framework allows a space for possible inputs {\em and} a space for possible solutions---the interplay between them seems to be the main source of combinatorial structures, and how well the two spaces are combinatorially related accounts for the complexity of the problem.
Some algorithms (e.g., that for $\Core\uu$) obtain their efficiency by directly shrinking the solution space. Even for those that shrink the input space (e.g., those for $\Sat\uu$ and $\nash\uu$), the key is to explore the relation of the two spaces and to design trials such that even a worst-case violation can be used to cut out a decent fraction of the input space.



%% file: Related-Work.tex
\subsection{Relation to Existing Work}\label{sectoin-related-work}


Our model with unknown inputs bears a resemblance to certain other problems and models, e.g., learning, algorithm design in unknown environments, ellipsoid method, and query complexity. However, there are fundamental distinctions between these models and ours. We now provide a detailed discussion of our work's relation to these models and problem.

\medskip
\noindent {\bf Learning.}
Our model has strong connections to various learning theories,
but fundamental differences also exist.

\begin{enumerate}
\item Learning theories, in essence, aim to identify the unknown object {\em itself}, either exactly (as in concept learning) or approximately (sometimes in the form of a prediction, as in PAC learning and active learning). In our model, however, the ultimate goal is quite different: we attempt only to find a solution of an unknown object, without necessarily learning the object itself. For certain applications (such as the aforementioned development of diagnostic reagents), finding a solution is indeed the main mission.

\item It is important to note that in our model, a solution may be found long before the exact input is learnt. Further, in certain cases, such as $\Sat\uu$ and $\nash\uu$, the exact input may take an exponential number of queries or even be impossible to learn.
    For $\Sat\uu$, even if we relax the requirement by allowing to output any formula within a cluster in which all formulas have the same set of satisfying assignments as the hidden input formula,
    it is still exponentially harder than finding a solution (Proposition~\ref{prop-SAT-learning}).
    Our algorithm, in contrast, is able to find a solution in polynomial time without learning the exact input formula.


\item Both similarities and differences abound in existing learning theories, and deciding which one to use in a specific application largely depends on the available method of accessing the unknown. As we have demonstrated, there are a fairly large number of scenarios in which the only available access to the unknown is provided by the verification oracle, but existing learning theories do not seem to address such situations.
\end{enumerate}
In summary, with its solution-oriented objective and advantages in computational efficiency, the present work is hopefully to serve as a useful supplement to existing learning theories, particularly in contexts in which the unknown object itself is impossible or unaffordable to learn and the only available access to the unknown is through a solution-verification process. 

Next we give a detailed discussion on the relationship between our model and relevant learning theories.

\begin{itemize}
\item {\em Concept learning}.
    In concept learning, a concept is secretly drawn from a given concept class, and the task at hand is to identify it; see the survey by Angluin~\cite{Ang04}. More precisely, given a domain set $X$ an      d a collection $C$ of concepts, each $c\in C$ maps $X\rightarrow \{0,1\}$, defining a table with rows $C$ and columns $X$. Further, there is an unknown concept $c^*\in C$. To identify the row/concept $c^*$, two types of queries are commonly used: (i) a {\em membership query}, where one queries a particular column/domain element $x\in X$, and an oracle returns its value $c^*(x)$,
    and (ii) an {\em equivalence query}, where one proposes a particular row/concept $c\in C$, and an oracle returns either a confirmative answer or a column/domain element $x$ with $c(x) \neq c^*(x)$. If the proposed concept $c$ is allowed to be any mapping from $X$ to $\{0,1\}$, then it is called an extended equivalence query. 

    Despite having some similar features, our model cannot be cast into the membership or equivalence query frameworks. For a search problem, we can take the set of inputs as the concept class $C$, take all possible solutions as the domain $X$, and let $c(x)=1$ if and only if $x\in X$ is a solution of the instance $c$. A membership query is thus similar to our model in that ours also proposes solutions $x\in X$. However, the goal in our problem is to find a solution $x$ that satisfies $c^*(x)=1$, where $c^*$ is the hidden input, whereas the task in the membership query model is to identify the concept $c^*$. The second important difference is that other than knowing whether or not a proposed solution is valid, in our model, we also gain extra information on the index of a violated constraint. Both this extra information and the feature of not necessarily learning the input enable us to obtain (more) combinatorial structures and to design efficient algorithms.

    Alternatively, we can take the solution space as $C$ and the set of constraints (corresponding to an input instance) as $X$; then, our problem is to find a solution $c$ that satisfies $c(x)=1$ for all $x\in X$, i.e., all constraints are satisfied. Our trial and error model then becomes similar to the equivalence query model: each time we propose a solution $c$ and an oracle returns a violated constraint $x$ where $c(x)=0$.
    However, note that different input instances correspond to different sets of constraints, and thus induce different $(C,X)$ tables. As a result, there are actually many such tables in our model, and we do not know which one we are faced with. Our task is still to find a row (i.e., a solution) in which all entries are 1 in the hidden table, and the algorithm needs to succeed for {\em any} of the possible input tables $(C,X)$. Thus, both the unknowns and objectives of the two models are entirely distinct.

    We could also put all of the possible constraints together to form a much larger table, with the rows indexed by all possible solutions $c\in C$ and the columns indexed by all possible constraints $x\in X$, with $c(x) = 1$ if the solution $c$ satisfies the constraint $x$, and $c(x) = 0$ otherwise. This approach would produce only one fixed table (as in concept learning), but the hidden input would now correspond to a hidden subset of columns. When a violation $x$ is returned, it is not the identity of a column of the large table, but the column's relative position inside the hidden subset of columns. Further, the required task is also different. In the equivalence query model a row is to be identified, whereas in our model we are asked merely to return one row $c\in C$ that satisfies $c(x) = 1$ for all $x$ from the hidden subset of columns.

    In addition, even more differences exist between our model and concept learning. In our model, a search problem instance may have several different solutions; thus, the correct output of an algorithm is not unique. However, in concept learning, the return of an algorithm has to be the unique hidden input $c^*$. Another difference lies in the complexity measure. In concept learning, the cost of an algorithm is measured in terms of $|C|$ and $|X|$. In our model, in contrast, the complexity is evaluated in terms of the input size of the search problem, rather than the size of the solution space, which can be exponential or even unbounded.

\item {\em Active learning}. A more general type of query learning model is active learning. Roughly speaking, there are two sets, $L$ and $U$, whose data are labeled and unlabeled, respectively. Based on labeled set $L$, a learning algorithm interactively queries an oracle concerning certain data instances from unlabeled set $U$. The oracle then returns the labels of the queried data, which are added into $L$. See the survey by Settles~\cite{Set09} for a more detailed discussion. There are a number of similarities between active learning and our model. For example, both consider an iterative and adaptive process that involves posing selected queries to a given oracle, and both investigate the complexity of interacting with this oracle~\cite{Das05,BHV10}. However, the two oracles are fundamentally different. In active learning, the oracle always returns the {\em correct} answer (i.e., label), whereas in ours it always returns one of the {\em erroneous} constraints. In addition, the objective in active learning is to achieve a high degree of accuracy (for the prediction of the data in unlabeled set $U$) using as few labeled instances as possible, which is very different from the objective in ours.

\item {\em PAC learning}. Another slightly related model is Valiant's probably approximately correct (PAC) learning model~\cite{Val84}, in which from a probability distribution $D$ over the domain set $X$, we can sample instances and be told whether or not they are supported by the unknown concept. Based on the training samples (the minimum number of which is called the sample complexity), the objective is to propose a concept (called a hypothesis) that approximates the unknown concept with a small probability of error (which can be used to predict future samples). Classic examples include DNF learning~\cite{Jac97,KS04,BMO05} and halfspace learning~\cite{KOS04,KKM08,KS09}. As in concept learning, PAC also aims to learn the unknown itself (although an approximation is allowed) rather than an induced solution, as in our model.


\end{itemize}

There are also many other learning models, e.g., decision tree learning, reinforcement learning~\cite{BS98}, statistical learning~\cite{Vap98}, (semi-)supervised learning~\cite{BB10}, and learning with errors~\cite{Re09}; see~\cite{KV94,Mit97} and the references therein for a more detailed discussion. The high-level philosophy of these models is also ``sample and predict'', 
which is very different from our trial and search (for a solution).

\medskip
\noindent {\bf Algorithms in an unknown environment.}
Theorists have addressed the exploration of an unknown environment in several domains.
In robot exploration and navigation, a robot is placed in an unknown geometric terrain with obstacles~\cite{BRS97,DKP98,HIK01,AKS02,IKL04} or a graph with unknown edges~\cite{DP90,PP98,ABR99,AH00},
and the robot's goal is to explore the entire unknown space starting from a given point.
A similar problem is path planning, where again we are given an unknown environment, but the objective is to find a desired path between two specific locations in either graphs~\cite{PY91} or geometric spaces~\cite{LS95,LS96,AWZ00}. Readers are referred to~\cite{Mit98} for a survey of these topics and results.
A feature our model has in common with these is that the input object under study is unknown.
However, as opposed to the proposal of entire solutions in our model, solution finding in the robot exploration and path planning models takes place via a ``local search'' type approach, where an agent looks for the next move on the basis of his or her current location and historical information. Despite being a very natural model for these applications, it seems difficult to generalize to scenarios without an underlying geometric structure. 

\medskip
\noindent {\bf Ellipsoid method.}
Our trial and error search model is, in spirit, similar to the ellipsoid method (see, e.g.,~\cite{GLS88}), in which a point is proposed as a trial, and a separating hyperplane is returned as an error. The ellipsoid method is an elegant approach for proving the polynomial time solvability of a class of combinatorial optimization problems; it applies even when the explicit expressions of the constraints are unknown.
However, our trial and error model includes a much broader class of search problems---not only convex optimization problems, but also many with pure combinatorial structures (e.g., the \Sat, \GpI, and \GraphI problems considered in our paper). From this perspective, the ellipsoid method is only one possible approach for the trial and error search problems in our model. (Indeed, the algorithms for the $\Core\uu$ and $\nash\uu$ problems considered herein are built crucially on the ellipsoid method; but we also employ other approaches to solve other problems, including $\Sat\uu$ and $\SM\uu$.) In addition, even if a problem can be solved using an ellipsoid-based approach, its trial and time complexities may be quite large (e.g., the ellipsoid method cannot compete with the simplex algorithm for practical calculations). Therefore, for problems with numerous applications, e.g., $\SM\uu$, a more efficient (combinatorial) algorithm is desirable (note that a stable matching instance can be written by a linear program~\cite{Van89}).

\medskip
\noindent {\bf Complexity.}
In the query model (also known as decision tree model), an algorithm makes queries in the form of ``$x_i = ?$'', and the task is to compute a function $f$ on the unknown $x$ by the minimum number of queries~\cite{BdW02}. 
Although this area has the same flavor of computing a function without learning all input variables, it is quite different from our model in the form of queries allowed.\footnote{Other forms of queries were also considered, such as those in linear decision trees and algebraic decision trees, but they are still in a very restricted form of queries.}
Therefore, our results on trial complexity can be viewed as an extension of the traditional query model by allowing a much larger class of queries with natural motivations.

In some cryptographic tasks, input instances are hidden from one party. For example, in instance-hiding proof systems~\cite{BFOS93}, the verifier tries to compute a function $f$ on input $x$ by interacting with one or more provers, without leaking any information on $x$ to the prover(s). Although this model is clearly very different from ours, an interesting research direction would be to explore the connections between our model and various cryptographic tasks.

%% file: Stable-Matching.tex
\section{Stable Matching}\label{section-stable-matching}

In a Gale-Shapley two-sided matching market model~\cite{GS62}, we are given a set of men $M$ and a set of women $W$, where $|M|=|W|=n$. Each man $m \in M$ has a strict and complete {\em preference list},\footnote{We follow the model proposed by Gale and Shapley in their seminal work~\cite{GS62}, where the number of men and women is the same, and every individual's preference is assumed to be complete and strict. All of these assumptions can be removed in our results, but for simplicity of exposition, we will adopt Gale and Shapley's original model.} denoted by $\succ_m$, ranking all the women in $W$, where $w_1\succ_m w_2$ means that $m$ prefers $w_1$ to $w_2$.
The preference list $\succ_m$ is assumed to be transitive (i.e., if $w_1\succ_m w_2$ and $w_2\succ_m w_3$, then $w_1\succ_m w_3$). The preference list $\succ_w$ of every woman $w\in W$ is defined similarly.

Given a matching between $M$ and $W$, denoted by $\mu$, we say that $m\in M$ and $w\in W$ form a \emph{blocking pair} if
both prefer each other to their matched partner in $\mu$, i.e., $w \succ_m \mu(m)$ and $m\succ_w \mu(w)$, where $\mu(m)$ and $\mu(w)$ are the woman matched to $m$ and the man matched to $w$ in $\mu$, respectively.
Matching $\mu$ is called {\em stable} if it contains no blocking pair. The $\SM$ problem is to find a matching that is stable, namely, a matching that satisfies the following set of constraints, labeled by man-woman pairs. 
\begin{equation}
\mbox{Either $\mu(m) \succ_m w$ or $\mu(w)\succ_w m$,\ $\forall m\in M, w\in W$.} 
\end{equation}
A stable matching always exists, and can be computed by Gale-Shapley's deferred acceptance algorithm in time $O(n^2)$.

In the unknown-input version of the stable matching problem, denoted by $\SM\uu$, we do not
know the preference lists $\succ_m$ and $\succ_w$. What we can do
is to propose candidate matchings as potential solutions. If a proposed matching is indeed stable, then the verification oracle $\V$ returns \yes, and the problem is solved. If it is not  stable, then one constraint (i.e., a blocking pair) is revealed by $\V$.
Our first result is an $O(n^2\log{n})$ upper bound on the randomized trial complexity of $\SM\uu$.

\begin{Thm}\label{Thm:SMUpperBound}
  There is a polynomial-time randomized algorithm solving $\SM\uu$ with~$O(n^2\log{n})$ trials.
\end{Thm}

Before describing the idea underlying the proof of this result, we first look at the
unknown-input version of another basic problem, \Sort, which has a close relationship
with \SM.
In $\Sort\uu$, there is a set of
$n$ elements $S = \{a_1, a_2, \dots, a_n\}$ in some underlying
linear (total) order $\succ$, but this order is unknown to
us, and the task is to discover it. We
can propose a linear order $(a_{k_1}, a_{k_2}, \ldots, a_{k_n})$ each time. If it is
indeed the desired hidden total order, i.e., $a_{k_1} \succ a_{k_2}
\succ \dots \succ a_{k_n}$, then the verification oracle returns
\yes, and the problem is solved; otherwise, a pair of elements $(a,
b)$ is returned such that $a$ is before $b$ in the proposed order, whereas
$b \succ a$ in the actual order.

It is well known that the time complexity of a comparison-based sorting problem is $\Theta(n\log{n})$. Note that this time complexity for \Sort is completely different from our trial complexity for $\Sort\uu$. In the following, we will show that the trail complexity bound for $\Sort\uu$ is actually also $\Theta(n\log{n})$.

\begin{Lem}\label{Lem:Sorting}
  $\Sort\uu$ can be deterministically solved using $O(n\log{n})$ trials, and the running time can be made polynomial for randomized algorithms. Meanwhile, any randomized algorithm that solves $\Sort\uu$ needs at least $\Omega(n\log{n})$ trials, even with unbounded computational power.
\end{Lem}

\noindent
{\em Idea of the proof.}
The upper and lower bounds for $\R(\Sort\uu)$ both use order theory~\cite{Bri99,DP02}. Both bounds critically depend on how fast the set of complete orders consistent with a partial order can be shrunk by \emph{worst-case} pair violations. Due to the distinction of the oracles in the standard comparison model and in ours, the techniques in the comparison model do not straightforwardly carry over to solving our problem. It turns out that controling the measure of \emph{average height} 
allows us to bound, in both directions, the worst-case shrinkage speed in our model. Although 
the average height is $\sharpp$-complete to compute \cite{BW91}, fortunately, we can efficiently estimate this value to a sufficient degree of precision 
by using a fully polynomial randomized approximation scheme (FPRAS) 
for another related problem \cite{DFK89,BD98}. \hfill $\square$
\medskip

\begin{proof}[Formal Proof of Lemma~\ref{Lem:Sorting}]
  First, we define the notation as follows. A \textit{partially
  ordered set} (or \textit{poset}) is a set $S$ equipped with an
  irreflexive transitive relation $>$. A \textit{linear extension} of
  a poset $(S, >)$ is a linear order $\succ$ over set $S$ such that $a
  \succ b$ whenever $a > b$ in $S$. For any given poset $(S, >)$
  and elements $a, b \in S$, we denote by $\Pr(a \succ b)$ the
  probability of $a \succ b$ where $\succ$ is chosen uniformly at random among linear extensions
  of $(S, >)$.

  In the $\Sort\uu$ problem, suppose the unknown order is $\succ^*$. Notice that for each of our proposed orders $\succ$, the verification oracle returns a pair of elements $a, b \in S$, from which and $\succ$, we can infer the relation between $a$ and $b$ in the actual order $\succ^*$. Thus, at each point of the algorithm, the information collected so far forms a poset $(S, >)$, of which the underlying unknown order $\succ^*$ is a linear extension.

  For any poset $(S, >)$, we call an order $(a_{k_1}, a_{k_2}, \dots,
  a_{k_n})$ \textit{good} if there is a constant $c<1$, such that for any $i < j$, we always have
  $\Pr(a_{k_j} \succ a_{k_i}) < c$. Note that this is a very strong condition because it requires a constant shrinkage for \emph{all} possible pairs $(i,j)$. The idea of solving $\Sort\uu$
  is that, at each step of the algorithm when our collected
  information forms a poset $(S, >)$, we propose a good order if it exists. The property of a good order guarantees that whatever the verification oracle returns, we can always reduce the number of candidate linear extensions by a constant fraction $c$. Note that at the beginning of the algorithm, the number of candidate linear extensions is $n!$ as we do not have any information about $\succ^*$. And at the end of the algorithm the unique order $\succ^*$ is found and thus the number of linear extensions is 1. Therefore, the problem $\Sort\uu$ can be solved in polynomial time and by $O\big(\log_{1/c}{n!}\big) = O(n\log{n})$ trials to the verification oracle, provided that 
  \begin{enumerate}
  \item for any poset $(S, >)$, a good order always exists, and
  \item we find a good order in polynomial time.
  \end{enumerate}

  Next we shall address how to satisfy these two conditions. Given a poset $(S, >)$ and an element $a \in S$, define its \textit{average height} $h(a)$ to be the average rank of $a$ in
  all linear extensions of $(S, >)$, where the rank of $a$ in a
  linear extension $\succ$ is the number of elements $b$ such that $b
  \succ a$. It is easy to see that the average height of elements in
  $S$ are all rational numbers between $0$ and $n-1$.
  Kahn and Saks~\cite{KS84} showed that for any pair of elements $a, b
  \in S$ satisfying $h(a) - h(b) < 1$, we must have $\Pr(b \succ a) <
  \frac{8}{11}$. Thus, for any poset $(S, >)$, if we sort all elements
  of $S$ in order $(a_{k_1}, a_{k_2}, \dots, a_{k_n})$ such that
  $h(a_{k_1}) \leq h(a_{k_2}) \leq \dots \leq h(a_{k_n})$, then for
  any $i < j$ we will have $h(a_{k_i}) - h(a_{k_j}) \leq 0 < 1$, and
  thus $\Pr(a_{k_j} \succ a_{k_i}) < \frac{8}{11}$. This implies that
  this is a good order as we want.

  Therefore, it remains to compute $h(a)$ efficiently for every element $a$.
  However, it was shown in~\cite{BW91} that counting the number
  of linear extensions of a given poset is \sharpp-complete, and
  determining the average height of an element of a poset is
  polynomially equivalent to the linear extension counting problem,
  thus is also \sharpp-complete. Luckily, to our purpose of having a good order, an approximation of $h(a)$ to a small enough precision suffices.
  To this end, we first notice that there exists a fully polynomial randomized approximation scheme (FPRAS) for the problem
  of counting the number of linear extensions~\cite{DFK89,BD98}, where the algorithm finds, with probability $1-\delta$, a $(1+\epsilon)$-approximation of the number of the
  linear extensions in time $poly(n,1/\epsilon,\log(1/\delta))$. Now given a poset $(S, >)$ and two elements $a, b \in S$, we can apply this algorithm to posets $(S,>)$, getting an output $n_1$, and apply the algorithm on $(S, >')$, getting an output $n_2$, where $>'$ is obtained from $>$ by incorporating an extra relation $(a > b)$. Then $\Pr(a \succ b)$ can be approximated by $n_2/n_1$ (with the precision $\frac{1+\epsilon}{1-\epsilon})$. It is easily seen from the definition of $h(a)$ and that of $\Pr(b \succ a)$ that
  $$h(a) = \sum_{b \neq a}\Pr(b \succ a).$$
  By setting $\epsilon = \frac{1}{5n}$ to approximate the value of each $\Pr(b \succ a)$ and using them to compute the value of $h(a)$, we can derive in polynomial time a value $h'(a)$ such that
  $1-\frac{1}{2n} < \frac{h'(a)}{h(a)} < 1 + \frac{1}{2n}$ with high
  probability. Since $h(a) < n$, we have
  $$|h'(a) - h(a)| < \frac{h(a)}{2n} < 0.5.$$
  Using $h'(a)$ to sort all elements in $S$ in order $(a_{k_1}, a_{k_2}, \dots, a_{k_n})$, we have by the above inequality that $h'(a_{k_i}) < h'(a_{k_j})$ for
  any $i < j$ with arbitrarily high probability. This implies $h(a_{k_i}) - h(a_{k_j})
  < 1$, and thus, $\Pr(a_{k_j} \succ a_{k_i}) < \frac{8}{11}$.
  Therefore, this is a good order for the current poset $(S, >)$.
  The whole process can be done in polynomial time, which completes the proof of the upper bound side.

  For the lower bound side, it was also shown in~\cite{KS84} that for
  any poset $(S, >)$, there exist two elements $a, b \in S$ such that
  $\Pr(a \succ b) > \frac{3}{11}$ and $\Pr(b \succ a) > \frac{3}{11}$.
  Thus, we construct the verification oracle as follows. At any step,
  when the previously returned information forms a poset $(S, >)$,
  no matter what the current proposed order is, the oracle always
  returns such $(a, b)$ (or $(b, a)$, depending on
  their relative position in the proposed order) as a violation. Then after this
  trial, at least $\frac{3}{11}$ fraction of the possible linear
  extensions still remains. Therefore, we have $\R(\Sort\uu) \geq
  \log_{11/3} n! = \Omega(n\log{n})$.
\end{proof}


Having the upper bound result for $\Sort\uu$, we can consider the preference of each individual as a sorting problem and solve these $2n$ $\Sort\uu$ problems together, which gives us the desired upper bound for the $\SM\uu$ problem. 

\begin{proof}[Proof of Theorem \ref{Thm:SMUpperBound}]
  At each step, since our collected information gives a poset $(W,>_m)$ for each
  man $m$ and a poset $(M,>_w)$ for each woman. We can use (part of) the
  above algorithm for $\Sort\uu$ to compute in polynomial time a good order $\succ_m$
  for each man $m$ and a good order $\succ_w$ for each woman $w$ with respect to their current posets. Then we
  take these good orders as their preferences and compute a stable
  matching $\mu$ using Gale-Shapley's algorithm. Then we make a query $\mu$ to the
  verification oracle for $\SM\uu$. If it is not a stable matching, then \V returns
  a blocking pair $(m,w)$ satisfying $w \succ_m \mu(m)$ and $m\succ_w \mu(w)$.
  Notice that at least one of these two new relations conflicts with the assumed
  preferences $\{\succ_m, \succ_w\}$, which means that it can serve a valid verification oracle for $\Sort\uu$.
  Since (i) there are totally $n$ man and $n$ woman,
  (ii) at each step at least one man or woman gets a new pair
  of elements to update their posets, and (iii) by Lemma~\ref{Lem:Sorting}
  we know that each poset will be updated at most $O(n\log{n})$ times,
  we know that this stable matching algorithm calls the verification oracle $2n \cdot O(n\log{n}) = O(n^2\log{n})$ times. Finally, since
  all computation between trials can be done in polynomial time, the
  time complexity of the algorithm is in polynomial.
\end{proof}


Note that the same results for deterministic algorithms hold; that is, there is an exponential-time deterministic algorithm using $O(n\log n)$ trials for $\Sort\uu$ (and thus another algorithm using $O(n^2\log n)$ trials for $\SM\uu$), namely,
$\D(\Sort\uu) = O(n\log n)$ and $\D(\SM\uu) = O(n^2\log n).$
The reason is simple: we can compute the average height $h(a)$ for every $a$ by enumerating all linear extensions of the current partial order. See Section \ref{sec:conclusion} for more discussions on polynomial time deterministic algorithms.

Note that our algorithm for $\SM\uu$ essentially runs $2n$ $\Sort\uu$ instances to learn the entire unknown input of the preference lists, and analysis of the trial cost simply involves adding the trials made on these instances. Although the $\Omega(n\log{n})$ lower bound for $\Sort\uu$ implies a limit to this approach, there seems to be considerable room for improvement. It may be unnecessary to learn all of the input lists to find a solution, and there may be more sophisticated ways to solve $2n$ instances of $\Sort\uu$ to beat the naive upper bound by addition. However, the following theorem gives an almost matching lower bound for $\SM\uu$.


\begin{Thm}\label{theorem-SM-lower}
Any randomized algorithm for $\SM\uu$ needs at least $\Omega(n^2)$ trials even with unbounded computational power.
\end{Thm}

\begin{proof}
For each man $m$, we create a directed graph $G_m$ with vertex set $W$ and edge set initially empty; similarly create a directed graph $G_w$ for each woman $w$. For each proposed matching $\mu$, after a blocking pair $(m,w)$ is
returned by the verification oracle, we know that $m$ prefers $w$ to
$\mu(m)$ and $w$ prefers $m$ to $\mu(w)$. We update the graphs $G_w$ and $G_m$ as follows. If there is no directed path from $\mu(m)$ to $w$, we add an edge $(\mu(m),w)$ in the graph $G_m$. Similarly, if there is no directed path from $\mu(w)$ to $m$, we add an edge $(\mu(w),m)$ in $G_w$. Since we have made $k$ queries to verification oracle, there are at most $k$ edges altogether in all graphs $G_m$ and at most $k$ edges in all graphs $G_w$. Note that these $2n$ graphs are all the information an algorithm gets from the trials and answers.

Our proof is by probabilistic method. Suppose we have already made $k$ queries.
We pick a pair $(m,w)$ uniformly at random in all possible $n^2$ pairs, and will show the following property for any matching $\mu$: If $k < (n^2-n)/2$, then with a strictly positive probability, $(m,w)$ is a blocking pair for $\mu$ on some input instance consistent with the queried information so far. Therefore, $k$ queries to the verification oracle are not sufficient to guarantee to find a stable matching, as for any candidate output $\mu$, there are still instances with blocking pairs for it.

The analysis of the probability goes as follows. First, with probability $1-1/n$, the pair is not in the current matching $\mu$. Second, we claim that with probability at least $1-k/n^2$, $m$
could prefer $w$ to $\mu(m)$. Here ``could" means that there is an input instance consistent with the previous trials and answers but in the input $m$ prefer $w$ to $\mu(m)$. Indeed, this could not happen only if
the known preferences of $m$ already imply that $w$ is less
preferred than $\mu(m)$; namely there is a path from $w$ to $\mu(m)$
in the graph $G_m$. But there are not many preferences
known---on average, only $k/n$ preferences known for $m$.
Formally, for a randomly chosen $w$, the event that ``$w$ is known
to be less preferred than $\mu(m)$ by $m$'' happens with
probability
\begin{eqnarray*}
    \frac{1}{n} \cdot \text{\big(the number of nodes that can reach $\mu(m)$ in $G_m$\big)} \leq  \frac{1}{n} \cdot \text{\big(the number of edges in $G_m$\big)}
\end{eqnarray*}
So averaging over all men gives that
\begin{eqnarray*}
    \Pr\big[m \text{ prefers } w \text{ to } \mu(m)\big] &\geq& 1-\frac{1}{n} \cdot \frac{1}{n} \cdot \text{\big(the total number of edges in all men's graphs\big)} \\
    &\geq& 1-k/n^2.
\end{eqnarray*}

Third, similar analysis shows that with probability at least
$1-k/n^2$, $w$ could prefer $m$ to her current assignment
$\mu(w)$. Putting all these together, with probability at least
$1-\frac{1}{n}-\frac{2k}{n^2}$, the pair $(m,w)$ can be a
blocking pair. The probability is strictly positive if $k <
(n^2-n)/2$, meaning the existence of a blocking pair. Therefore $k$
needs to be larger than $(n^2-n)/2$.
\end{proof}


A final comment is that in the traditional stable matching problem, where the preferences are known, there is a tight lower bound $\Omega(n^2)$ for computing a stable matching~\cite{Ng89}. However, this bound refers to the computational time complexity in a completely different meaning from our trial complexity lower bound.

%% file: SAT.tex
\section{SAT}\label{section-SAT}

\newcommand{\AlgSAT}{{\sc Alg-SAT}\xspace}

Given a CNF formula $\phi$ with $n$ variables and $m$ clauses, the $\Sat$ search problem is to find a satisfying assignment to $\phi$ if one exists, or to return ``$\phi$ is unsatisfiable.'' In the unknown-input version $\Sat\uu$, the formula $\phi$ is unknown, and each time that a proposed solution $x$ is not a satisfying assignment, the verification oracle returns an index $i$ such that the $i$-th clause of $\phi$ evaluates to FALSE on $x$.

First, we clarify that the computation oracle for $\Sat$ is defined as follows. On a query $\phi$ which is a CNF formula, $\Sat$ returns a satisfying assignment $x$ to $\phi$, or reports that such $x$ does not exist. So this $\Sat$ oracle solves the search problem instead of the decision problem. This is for the fair comparison since the target $\Sat\uu$ is also a search problem. Also note that the search and the decision problems for $\Sat$ are roughly the same due to the standard self-reducibility. 

Our algorithm solving the $\Sat\uu$ problem is as follows.

\begin{algorithm}[H]
  \caption{\AlgSAT}
  {\bf Unknown input}: A formula $\phi$ of $n$ variables and $m$ clauses.

  \begin{algorithmic}[1]
  \small
    \STATE Let $L_1 = L_2 = \cdots = L_m = \{x_1, \bar x_1, \ldots, x_n, \bar x_n\}$.

    \LOOP
    \STATE Let $\phi' = \bigwedge_{i=1}^m (\vee_{\ell\in L_i} \ell).$
    \STATE Query the computation oracle $\Sat$ on $\phi'$. 
    \IF {the computation oracle says that $\phi'$ is not satisfiable,}
        \RETURN $\phi$ is unsatisfiable (and terminate the program).
    \ELSE
        \STATE Suppose the computation oracle gives a satisfying assignment $x$ of $\phi'$.
        \STATE Ask the assignment $x$ to the verification oracle $\V$.
        \IF {$\V$ confirms that $\phi(x)=1$}
            \RETURN $x$ (and terminate the program).
        \ELSE
            \STATE Suppose $\V$ returns an index $i$.
            \STATE Let $L_i \leftarrow \{\ell\in L_i: \ell(x) = 0\}$ \\ (i.e., those literals in $L_i$ evaluating FALSE on $x$).
        \ENDIF
    \ENDIF
    \ENDLOOP
  \end{algorithmic}
\end{algorithm}

The idea of the algorithm uses a standard candidate elimination approach~\cite{Mit97}, which has been used extensively in learning theory (e.g., PAC learning for CNF formulas~\cite{Val84}).
The algorithm initially includes all literals, i.e., $x_1\vee \bar x_1\vee \cdots\vee x_n \vee \bar x_n$, in each clause. It proceeds by employing the $\Sat$ computation oracle to propose an assignment $x$ consistent with the current knowledge of the clauses. If a clause's index is returned by the verification oracle upon a trial, then we know that $x_i$ cannot be in the clause if $x_i=1$ and that $\bar x_i$ cannot be in the clause if $x_i=0$ in the assignment of the trial. We therefore remove the literals from the clause, and continue the process until either a satisfying assignment is found or the computation oracle returns the result that no satisfying assignment exists.


\begin{Thm}\label{theorem-SAT-upper}
\AlgSAT solves the $\Sat\uu$ problem in polynomial time using $O(mn)$ trials, where $n$ and $m$ are the numbers of variables and clauses, respectively.
\end{Thm}

 \begin{proof}
   \emph{Correctness}. The set $L_i$ maintains a collection of possible literals for the $i$-th clause in the unknown formula $\phi$. Note that if $\phi$ is satisfiable, so is $\phi'$, at any step of the algorithm. Indeed, at the beginning all literals are included in each clauses and $\phi'$ is trivially satisfiable. Each time $L_i$ is updated, the literals that are removed from $L_i$ are exactly those $\ell$ with $\ell(x) = 1$. But the literals in the actual $i$-th clause $C_i$ of $\phi$ all evaluate to 0, because $C_i(x) = \vee_{\ell\in C_i} \ell(x) = 0$. Thus, all the literals in $C_i$ are kept when $L_i$ is updated. Therefore a satisfying assignment $x$ to $\phi$ also satisfies $\phi'$.

   Also note that if $\phi'(x) = 1$, then there exists at least one $\ell\in L_i$ such that $\ell(x) = 1$. So the size $|L_i|$ decreases by at least 1 each time $L_i$ is updated. Since $\phi'$ is always satisfiable and $\sum_i|L_i|$ decreases by at least 1 in each iteration, the algorithm finally stops and outputs an $x$ with $\phi(x) = 1$.
   On the other hand, if $\phi$ is unsatisfiable, then $\phi(x)$ never evaluates to 1, thus finally the algorithm outputs that $\phi$ is unsatisfiable.

   \emph{Complexity}. Since initially $|L_1| = \dots = |L_m| = 2n$, the algorithm runs in at most $2nm$ rounds. In each round finding a satisfying assignment of $\phi'$ takes one query to the computation oracle, so the trial complexity are $O(nm)$. The time complexity of the algorithm is $poly(n)$ as all computations between queries can be done in polynomial time.
 \end{proof}

Our upper bound result implies that not knowing the input formula does not add much extra complexity (up to a polynomial) to finding a satisfying assignment.
For a related problem, 2$\Sat$, in which every clause contains at most 2 literals, which can be solved in polynomial time, we can show that solving 2$\Sat\uu$ in polynomial time implies $\p=\np$.\footnote{The idea of the reduction is similar to the one described in the next section for group isomorphism.} Therefore, it is indeed the power of the computation oracle $\Sat$ that yields our efficient algorithm for $\Sat\uu$.
Recall that the focus of our study is not to discover the complexity of solving $\Sat\uu$ itself, but rather, is to understand the relative complexity of solving $\Sat\uu$ compared with that of solving $\Sat$.

Note that our algorithm directly shrinks the space of possible inputs (i.e., formulas), instead of the space of possible solutions (i.e., assignments).
(While some variables may have their values fixed along the process of the algorithm, this may not necessarily be the case and is surely not the reason why our algorithm works within $O(mn)$ queries.) When a satisfying assignment is found by the algorithm, however, we may still do not know the exact input formula.
This echoes the medical treatment application mentioned in Introduction, where finding a solution is much more urgent than learning the unknown input, and it is indeed often possible to find an effective diagnostic reagent long before completely knowing the virus.

In the complexity language, we ask whether it is computationally much harder to find out the hidden CNF formula (even with the help of the computation oracle $\Sat$) than finding a solution. Learning the exact formula is clearly impossible since, for example, if there is a clause $x_i\vee \bar x_i$, then we can never know this clause. Even if we relax the requirement by clustering the formulas with the same set of satisfying assignments, and allowing to output any formula within the cluster of the hidden input formula, it is still exponentially harder than finding a solution. This separates ours from the standard concept learning model.

\begin{Prop}\label{prop-SAT-learning}
	Any randomized algorithm with $\epsilon$-error, $\epsilon < 1/2$, needs $2^n$ trials to output a formula $\phi'$ with the same set of satisfying assignments as the hidden input formula $\phi$.
\end{Prop}

\begin{proof}
	We will give a collection of $2^n + 1$ formulas as follows. For each $y\in\Bn$, define $\phi_y$ to be a formula of only one clause, which has $n$ literals. The $i$-th literal is $\bar x_i$ if $y_i = 1$, and is $x_i$ if $y_i = 0$. This makes $\phi_y$ to be satisfied by all assignments of $x$ except for $x = y$.  The last formula $\phi_\emptyset$ can be any tautology. We will confine the hidden input formula to within this selection $\big\{\phi_y ~|~ y\in \Bn\big\}\cup \{\phi_\emptyset\}$. Note that all these formulas have different sets of satisfying assignments; thus, outputting a formula $\phi'$, with the same set of satisfying assignments as $\phi$, requires to identify the input formula from this selection.
	
	Now for any randomized algorithm that tries to output a formula $\phi'$, and any first $k$ queries $x^1, ..., x^k$ it makes, we, as an adversarial oracle, always answer \yes. The only information that the algorithm obtains at this point is that the hidden input formula is not $\phi_{x^1}$, ..., $\phi_{x^k}$. If the algorithm outputs a formula, then the success probability would be $1/(2^n+1-k)$. Thus the algorithm cannot pin down $\phi$ with an error probability smaller than $1/2$ until asking all potential solutions in $\Bn$.
\end{proof}




We also have a lower bound on the randomized trial complexity of $\Sat\uu$, which matches the upper bound at least when $m = \Omega(n^2)$.

\begin{Thm}\label{theorem-SAT-lower}
If $m = \Omega(n^2)$, any randomized algorithm for the $\Sat\uu$ problem needs at least $\Omega(mn)$ trials, even with unbounded computational power. If $m = o(n^2)$, the lower bound is $\Omega(m^{3/2})$.
\end{Thm}

The proof of the lower bound is combinatorially involved, in which we construct instances such that only one candidate literal can be eliminated in the foregoing process. This is enforced with the help of certain clauses that confine all satisfying assignments to a very special form.
Further, we remark that the proof can be easily adapted to show the same lower bound for the decision problem, namely to decide whether a formula (with unknown clauses) is satisfiable.

\begin{proof}
We will first show the lower bound of $\Omega(n^3)$ for the special case when $m=\Theta(n^2)$. Then we will consider the cases when $m = \omega(n^2)$ and $m = o(n^2)$.
	
Let $m^* = \binom{n/3}{2}+\frac{n^2}{9}+3$. We will construct a family of CNF formulas with $n$ variables and $m$ clauses, where $m^*\leq m \leq O(n^2)$.
Without loss of generality, we assume that 3 divides $n$. Divide $[n]$ into three equal blocks: $B_1 = \big\{1,2,\ldots,\frac{n}{3}\big\}$, $B_2 = \big\{\frac{n}{3}+1, \frac{2n}{3}+2, \ldots, \frac{2n}{3}\big\}$, and $B_3 = \big\{\frac{2n}{3}+1, \frac{2n}{3}+2, \ldots, n\big\}$. We, as an adversary constructing the verification oracle, maintain a set $T$ of triples $(i_1, i_2, i_3)$, with $T$ initially containing $(i_1,i_2,i_3)$ with all $i_1\in B_i$, $i_2\in B_2$ and $i_3\in B_3$.
Each triple $(i_1, i_2, i_3)$ represents an assignment, denoted by $x(i_1, i_2, i_3)$, where $x_i=1$ if and only if $i\in \{i_1,i_2,i_3\}$. We will show that some of the formulas in the constructed family have a unique satisfying assignment of the form $x(i_1, i_2, i_3)$, which can only be found by at least $(\frac{n}{3})^3$ queries for any randomized algorithm and a carefully designed verification oracle.

The hidden formula $\phi$ has $\Theta(n^2)$ clauses, divided into two parts, each with $\Theta(n^2)$ clauses. The first part ensures that any satisfying assignment of $\phi$ needs to have exactly one variable $x_i = 1$ in the last block. This can be done by $\Theta(n^2)$ clauses because we use one clause $\vee_{i\in B_3} x_i$ to enforce that there is at least one $x_i = 1$ in that block, and then use $\bar x_i \vee \bar x_j$ to enforce that at most one of $x_i$ and $x_j$ is assigned to be 1. Thus, this part of formula
\[(\vee_{i\in B_1} x_{i}) ~ \bigwedge ~ (\vee_{i\in B_2} x_{i}) ~ \bigwedge ~ (\vee_{i\in B_3} x_{i})\bigwedge_{i,j\in B_3: i\neq j} (\bar x_i \vee \bar x_j) \]
ensures that a satisfying assignment has at least one $x_i = 1$ in each of the first two blocks, and exactly one $x_i = 1$ in the last block. This part has $\binom{n/3}{2}+3 = \Theta(n^2)$ clauses.

The second part of the formula consists of $\frac{n^2}{9}$ clauses, indexed by the pairs $(i_1,i_2)$ for each $i_1\in B_1$ and $i_2\in B_2$. The clause associated with $(i_1,i_2)$ is either $(\bar x_{i_1} \vee \bar x_{i_2})$ or $(\bar x_{i_1} \vee \bar x_{i_2} \vee x_{i_3})$ for some $i_3\in B_3$. For an arbitrary randomized algorithm, on each query $x$, if $x$ does not satisfy the first part of formula, then we return the corresponding clause in that part. (This first part can be even revealed to the algorithm for free.) Now assume that $x$ satisfies the first part and it has three positions $i_1, i_2, i_3$ being 1, one in each block. (If there are multiple 1's in the first two blocks, let $i_1$ and $i_2$ denote the positions of the first 1's in the corresponding blocks.) Let the verification oracle return the clause with index $(i_1,i_2)$; note that this implies that the clause cannot be of the form $(\bar x_{i_1} \vee \bar x_{i_2} \vee x_{i_3})$. We then update the candidate set $T$ by removing one element $(i_1,i_2,i_3)$ from it. Note that this is also all the information the algorithm obtains from this query, that is, for any $(i_1', i_2', i_3')\in T$,  $(i_1', i_2', i_3')\neq (i_1, i_2, i_3)$, $x(i_1', i_2', i_3')$ can still be a possible satisfying assignment.
	
This process continues; we claim that as long as $|T|\geq 2$, namely it contains at least two distinct triples $(i_1, i_2, i_3)$ and $(i_1', i_2', i_3')$, then the formula can be either with a unique satisfying assignment $x(i_1, i_2, i_3)$ or with a unique satisfying assignment $x(i_1', i_2', i_3')$, thus the algorithm does not know what to output yet. Indeed, all previous queries by the algorithm cannot distinguish the following two possibilities: The second part of the formula can be either
\[\phi_1 = \bigwedge_{(j_1,j_2) \neq (i_1,i_2)} (\bar x_{j_1} \vee \bar x_{j_2}) ~ \bigwedge ~ (\bar x_{i_1} \vee \bar x_{i_2} \vee x_{i_3})\]
\text{or}
\[\phi_2 = \bigwedge_{(j_1,j_2) \neq (i_1',i_2')} (\bar x_{j_1} \vee \bar x_{j_2}) ~ \bigwedge ~ (\bar x_{i_1'} \vee \bar x_{i_2'} \vee x_{i_3'})\]
where $j_k$ is in block $B_k$.
	
We claim that both formulas are satisfiable, and actually each has a unique satisfying assignment, namely $x(i_1, i_2, i_3)$ and $x(i_1', i_2', i_3')$, respectively. Let us consider $\phi_1$ as an example. First, it is easy to see that $x(i_1, i_2, i_3)$ does satisfy $\phi_1$.
Second, the assignment $x(i_1, i_2, i_3)$ is the only satisfying assignment. Suppose $x$ satisfies $\phi_1$, then to satisfy $(\vee_{i\in B_1} x_{i})\bigwedge (\vee_{i\in B_2} x_{i}) \bigwedge_{(j_1,j_2) \neq (i_1,i_2)} (\bar x_{j_1}\vee \bar x_{j_2})$, we have $x_{i_1}=x_{i_2} = 1$ and $x_i = 0$ for $i\in B_1\cup B_2 - \{i_1,i_2\}$. Otherwise, if $x_{j_1} = 1$ for any $j_1\neq i_1$ in $B_1$, then $x_{j_2} = 0$ for all $j_2\in B_2$, violating $\vee_{i\in B_2} x_{i}$. Thus, $x_{j_1} = 0$ for all $j_1\neq i_1$ in $B_1$; the clause $(\vee_{i\in B_1} x_{i})$ then forces $x_{i_1} = 1$. Similarly we can show that $i_2$ is the only position in $B_2$ with assignment 1. Now to satisfy the last clause in $\phi_1$, $x_{i_3}$ must be 1, and the first part of formula guarantees that $i_3$ is the only position $i$ in $B_3$ with $x_i = 1$. Thus, $x(i_1, i_2, i_3)$ is the only satisfying assignment.
	
Since initially $|T| = (n/3)^3$ and each query decreases $|T|$ by 1, we know that the algorithm needs at least $(n/3)^3 = \Omega(n^3)$ queries in the worst case.
	
Now consider the case that $m = \omega(n^2)$. Suppose $c > 2$ is the minimum integer that $m \leq (\frac{n}{2c})^c$. Divide the variables into $c+1$ blocks, with the first $c$ blocks of size $k_1 = (\frac{m}{2})^{1/c} < \frac{n}{2c}$, and the last block of size $k_2 = n - ck_1$. Note that $k_1 c \leq \frac{n}{2}$ and thus $k_2 \geq \frac{n}{2}$. By a similar setting of the previous formulas, we can use $\binom{k_2}{2}+1$ clauses to ensure that all the satisfying assignments have exactly one $x_i=1$ in the last block, and $c$ clauses to ensure that there is at least one $x_i=1$ in each of the first $c$ blocks. We also use $k_1^c$ clauses to hide a unique tuple $(i_1, \ldots, i_{c+1})$, with each $i_\ell$ in block $\ell$, such that the assignment $x$ with $x_i = 1$ if and only if $i = i_\ell$, for $\ell\in [c+1]$, is the unique satisfying assignment. Then we use $k_1^c + \binom{k_2}{2}+1+c < m$ clauses to hide the satisfying assignment, to find which needs $k_1^c k_2 \geq \frac{m}{2}\cdot \frac{n}{2} = \Omega(mn)$ queries.
	
Finally, for the case of $m < m^*$, we only use $\Theta(\sqrt{m})$ variables. The problem is then reduced to the first case, which gives a lower bound of $\Omega(m\sqrt{m}) = \Omega(m^{3/2})$.
\end{proof}

%% file: Group-ISO.tex
\section{Group Isomorphism}\label{section-group-iso}

Given two groups, $G$ and $G'$, of the same size, the group isomorphism (\GpI) problem is to find an isomorphism between $G$ and $G'$ or to report that ``$G\ncong G'$''. More precisely, given two groups, $G$ and $G'$, by their multiplication tables, $T_{n\times n}$ and $T'_{n\times n}$, our task is to output a bijection $\pi: G\rightarrow G'$ such that 
\begin{equation}
	\pi(a\circ b) = \pi(a)\circ' \pi(b) 
\label{eq:GroupIDef}
\end{equation}
for all $a,b\in G$ (where $\circ$ and $\circ'$ are the multiplications of $G$ and $G'$, respectively), or to report that ``$G\ncong G'$''.

Whether a polynomial time algorithm exists for the general $\GpI$ problem is a long-standing open question. Compared with another well-known problem, Graph Isomorphism, however, $\GpI$ has more group structures for potential use, and indeed, $\GpI$ can be solved in polynomial time if the given groups are Abelian, and (the decision version of) $\GpI$ is in $\np\cap \co\np$ (under certain complexity assumptions) if the given groups are solvable \cite{AT11}.

The problem is by nature a constraint satisfaction problem, searching for a bijection $\pi$ that satisfies the $n^2$ constraints Eq.\eqref{eq:GroupIDef}. In the unknown-input model, we can consider the case of both groups being unknown and that of exactly one group, say, $G$, being unknown. In either case, we can propose bijections $\pi$ to the verification oracle \V. If $\pi$ is not an isomorphism, then \V gives a pair $(a,b)$ with the foregoing equality violated. Our upper bound result in this section applies to the general case in which both groups are unknown, and our lower bound result applies even to the restricted case in which one group is known. 

A very natural attempt to solve $\GpI\uu$ is to keep proposing bijections $\pi$ that are consistent with our current knowledge of the restrictions of a valid bijection. More precisely, whenever \V returns a pair of elements $(a, b)$ on a proposed $\pi$, it adds the restriction that we cannot \emph{simultaneously} map the three elements $(a, b, a\circ b)$ to $(\pi(a), \pi(b), \pi(a)\circ' \pi(b))$. Thus, there are no more than $O(n^6)$ forbidden rules. Note that finding a bijection that avoids a list of forbidden pairs of triples is easy with an $\np$ computation oracle. Thus,
$\GpI\uu$ can be solved using $\Sat$ as the computation oracle in polynomial time and via $O(n^6)$ trials.


An immediate question that arises from this claim is the following. Does it hold with only a computation oracle $\GpI$ instead of $\Sat$? (After all, only comparison with $\GpI$ reveals the extra difficulty arising from the unknown input on the problem.) This seems quite conceivable that this is the case, as group theory by definition handles triples in the form $(\pi(a), \pi(b), \pi(a)\circ' \pi(b))$, and we have not yet exploited the group structures in the given tables. 
Surprisingly, this intuition turns out to be misleading, as shown by the following hardness result, which stands even if one group $G'$ is known to us and is a very simple group $\mbZ_p$ for a prime $p$. 
Denote the known-input version of this problem by $\GpIp$. Note that because it is solvable in polynomial time, the computation oracle is not needed (modulo a polynomial factor in runtime) in the unknown-input model. 

\begin{Thm}\label{theorem-group-iso}
  If\, $\GpIp\uu$ can be solved in polynomial time, i.e., $\GpIp\uu\in \p^{\scriptsize {\V}}$, then $\p=\np$. More specifically, if $\GpIp\uu$ can be solved in time $t(p)$, then $\HC$ can be solved in time $O(t(p) \cdot p)$, where $p$ is the order of the given group $\mbZ_p$.
\end{Thm}

\noindent
{\em Idea of the proof.}
The hardness result is shown by a reduction that is not very standard. Given a graph $H$ with $p$ vertices, we employ an algorithm $\mcA$ for $\GpI(\cdot,\mbZ_p)\uu$ to find a Hamiltonian cycle in $H$ in the following way. Assuming the existence of a Hamiltonian cycle $C$ (it can be seen that the primality of the size of the graph and the assumption of the existence of one Hamiltonian cycle do not alter the hardness of the Hamiltonian cycle problem), define a group $T$ via $C$ as follows.
Let $(b_0, b_1, \dots b_{p-1},b_0)$ be a Hamiltonian cycle $C$ in graph $H$.
We can fix a vertex $a$ and assume that $b_1 = a$, which
is always achievable by a cyclic shift in the labels in the Hamiltonian cycle if
necessary. For simplicity, we use the vertices of $H$ to denote the elements of group $T$. Now, for any $b_i$ and $b_j$ in group $T$, define their
multiplication by $b_i \circ b_j = b_{i + j \text{ mod }p}$. It is easy to see that $T$ is a
cyclic group with $p$ elements (where $b_1$ is a generator of the group).

The main idea of the reduction is to run algorithm $\mcA$ on input $(T,\mbZ_p)$, 
and translate the output of $\mcA$, which is an isomorphism
from $T$ to $\mbZ_p$, to a Hamiltonian cycle in the given graph $H$.
However, one problem immediately arises: Because the reduction algorithm has only polynomial time, and it cannot find such a Hamiltonian cycle, thus cannot construct the multiplication table $T$.

To get around this issue, we employ the crucial fact that $\mcA$ does not know its first input---all of $\mcA$'s information about $T$ comes from interactions with its verification oracle $\V$. Thus, it is sufficient to construct a $\V$ that answers $\mcA$'s trials. However, doing so again requires the information on $T$, which is exactly what we do not have. Here, the idea is to efficiently construct a \emph{simulator} $\V'$ to take the place of $\V$. Given the shortage of running time, it is inevitable that we lose something in our simulator $\V'$, and, in our final construction it turns out to be the correctness, which is the seemingly the most critical component. In other words, $\V'$ cannot answer all of $\mcA$'s questions correctly. What makes it still qualified for our purpose is the following key property. On any $\pi$ proposed by $\mcA$, $\V'$ 
\begin{enumerate}
\item either provides a correct response to $\pi$, or 
\item finds a Hamiltonian cycle in $H$. 
\end{enumerate}
This property means that the first time $\V'$ gives a wrong answer to $\mcA$, it has just found a Hamiltonian cycle in $H$. ($\mcA$'s output is admittedly now out of control now, but we no longer care about the correctness of $A$; we have used part of $\mcA$'s code to solve our $\HC$ problem.) \hfill $\square$

\begin{proof}[Formal Proof of Theorem~\ref{theorem-group-iso}]
 In the $\HC$ problem, we are given a graph $G$ with $p$ vertices and are asked if it has at least one Hamiltonian cycle\footnote{The primality of the graph size does not change the hardness of this problem. For a given graph with $n$ vertices for a general number $n$, one can first find a prime in $[n,2n]$ (which takes time $n\cdot poly(\log n))$ and then for an edge $(u,v)$ in the given graph $G$, add a path from $u$ to $v$ with $(p-n)$ extra vertices to form a new graph $G'$. It is easy to see that $G'$ has a Hamiltonian cycle if and only if $G$ has a Hamiltonian cycle using the edge $(u,v)$. Fix an arbitrary $u$, try all neighbors $v$ in $G$ using the algorithm for the new instances and verify the solutions, we will know that whether $G$ has a Hamiltonian cycle.}.
  Suppose that there exists an algorithm $\mcA$ that solves the $\GpIp\uu$ problem in time $t(p)$.
We now construct an algorithm $\mathcal{B}$ for the following variation of the $\HC$ problem: Given a graph $G$ with $p$ vertices and the condition that it contains at least one Hamiltonian cycle, find a Hamiltonian cycle in $G$; call the problem $\PHC$.

  Now we are given a graph $G$ with $p$ nodes, we label the nodes in $G$ as $a_1, a_2, \dots, a_p$ in an arbitrary way. We will construct an algorithm $\mcB$ that uses algorithm $\mcA$ to find a Hamiltonian cycle in $G$, assuming one exists.

  First we describe a cyclic group $T$\footnote{Precisely, $T$ is defined according to $G$ and should be written as $T(G)$. As all our discussions are with respect to the given graph $G$, for simplicity we will use $T$ in the proof.} with elements $a_1, \dots, a_p$. Let $(b_0, b_1, \dots b_{p-1})$ be a Hamiltonian cycle in graph $G$. If there is more than one Hamiltonian cycle, pick an arbitrary one. We can further impose that $b_1 = a_1$, which is always achievable by a cyclic shift if necessary, since a Hamiltonian cycle contains every vertex. Now for any $b_i$ and $b_j$ in group $T$, define their multiplication by $b_i \circ b_j = b_{i + j \text{ mod }p}$ (all additions hereafter are module over $p$). It is easy to check that $T$ is a cyclic group with $p$ elements.

  The algorithm $\mcB$ basically runs the algorithm $\mcA$ on input $(T,\mbZ_p)$. If finally $\mcA$ outputs a correct bijection $\pi$ mapping $T$ to $\mbZ_p$, then we can identify all $b_i$'s and thus find a Hamiltonian cycle. However, we do not know how to provide a valid verification oracle, because in polynomial time we cannot find $(b_0, b_1, ..., b_p)$ and define the multiplication table of $T$ as above. The idea here is to construct a simulator $\V'$ of the verification oracle in such a way that $\V'$
  \begin{enumerate}
  \item either provides correct responses to $\mcA$'s trials, or
  \item finds a Hamiltonian cycle for $\mcB$.
  \end{enumerate}
	
  The $\V'$ and the algorithm $\mcB$ are given below.
  \begin{algorithm}
    \caption{Simulator $\V'$ given $(G,\pi)$}
    \mbox {{\bf Input}: Graph $G$, and bijection $\pi: T \rightarrow \mbZ_p$} \\[-0.15in]
    \begin{algorithmic}[1]
      \STATE Let $x = \pi(a_1)$.
      \IF{$x = 0$}
      	\RETURN $(a_1, a_1)$.
      \ELSE
      	\IF{$\exists i\in \mbZ_p$ \st $(\pi^{-1}(ix), \pi^{-1}((i+1)x)) \notin E(G)$}
      		\RETURN $(\pi^{-1}(ix),a_1)$ for the first such $i$.
      	\ELSE
      		\RETURN \yes.
      	\ENDIF
      \ENDIF
		\end{algorithmic}
	\end{algorithm}
	
  \begin{algorithm}
  \caption{Algorithm $\mathcal{B}$ for the \PHC\ problem}
  {\bf Input}: \mbox{Graph $G$ with at least one Hamiltonian cycle} \\[-0.15in]
  \begin{algorithmic}[1]
    \STATE Suppose there is a Hamiltonian cycle $(b_0, b_1, ..., b_{p-1})$; circularly shift the cycle to make $b_1 = a_1$.
    \STATE Define a group $T$ with the multiplication table given by $T(b_i,b_j) = b_{i+j \text{ mod }p}$.
    \STATE Run $\mathcal{A}$ on input $(T,\mbZ_p)$, during which: 
    \IF {$\mcA$ makes a query $\pi$ to the verification oracle $\V$}
    	\STATE run $\V'(G,\pi)$ to simulate $\V$ to give either \yes or a pair $(a_i,a_j)$ as an answer.
     	\IF {the answer is \yes}
     		\STATE $x = \pi(a_1)$.
     		\RETURN a Hamiltonian cycle found: \\ $\big(\pi^{-1}(0), \pi^{-1}(x), \pi^{-1}(2x), \dots, \pi^{-1}((p-1)x)\big)$.
     	\ENDIF
    \ENDIF
    \IF {$\mcA$ outputs $\sigma$}
     	\RETURN a Hamiltonian cycle found: \\ \mbox{$\big(\sigma^{-1}(0), \sigma^{-1}(\sigma(a_1)), \sigma^{-1}(2\sigma(a_1)), \dots, \sigma^{-1}((p-1)\sigma(a_1))\big)$.}
    \ENDIF
  \end{algorithmic}
  \end{algorithm}

Several explanations of the algorithms are in order. First, note that in the course of the algorithm, $\mcB$ defines $T$ and runs $\mcA$ on input $(T,\mbZ_p)$. However, as we have mentioned, $\mcB$ actually does not know how to find a Hamiltonian cycle in polynomial time and to define $T$. Here we use the key property that $\mcA$ \emph{does not} know its input: by saying ``$\mcB$ runs $\mcA$'', we mean to let $\mcB$ run $\mcA$'s code between trials; whenever $\mcA$ makes a trial $\pi$, $\mcB$ uses $\V'$ to simulate the true verification oracle $\V$ to give an answer.
	
	This immediately raises the second issue: Our designed $\V'$ is not a valid verification oracle for the $\GpIp\uu$ problem, since it may return \yes for some wrong bijection. (That is, even if a bijection $\pi$ does not really map $T$ to $\mbZ_p$, our oracle $\V'$ may say \yes.) But what we can guarantee are the following two properties of $\V'$. The first is roughly the soundness for the algorithm $\mcA$.

	\begin{Claim}\label{claim:GroupIvio}
		If\, $\V'$ returns a violation $(a_i,a_j)$ to a proposed $\pi$, it is indeed a violation.
	\end{Claim}
	\begin{proof}
	If $\V'$ returns $(a_1, a_1)$ in the outer \textbf{if} statement, then it is a violation because, as $\pi$ is a bijection,
	\[\pi(a_1\circ_T a_1) = \pi(b_2)\]
    but
    \[\pi(a_1) + \pi(a_1) = 0+0 = 0 = \pi(a_1) = \pi(b_1) \neq \pi(b_2).\]
    (Here we use $\circ_T$ to emphasize that it is a multiplication of group $T$).
	
	The other possibility is that $\V'$ returns $(\pi^{-1}(ix), a_1)$: Suppose it is not a violation, then
	\begin{align}
		& \ (\pi^{-1}(ix),a_1) \text{ is not a violation} \\
		\Rightarrow & \ \pi(\pi^{-1}(ix) \circ_T a_1) = \pi(\pi^{-1}(ix))+\pi(a_1) \\
		\Rightarrow & \ \pi(\pi^{-1}(ix) \circ_T a_1) = ix+x \quad\quad\quad\quad (\pi(a_1) = x) &  \\
		\Rightarrow & \ \pi^{-1}(ix) \circ_T a_1 = \pi^{-1}((i+1)x) \quad\quad (\text{taking } \pi^{-1}) & \\
		\Rightarrow & \ (\pi^{-1}(ix),\pi^{-1}((i+1)x)) \in E
	\end{align}
	where the last line is because $\circ_T \ b_1 $ is defined as going to the next vertex along the Hamiltonian cycle $(b_0, b_1, ..., b_{p-1})$.
	\end{proof}

	The second property is roughly the soundness for $\mcB$; it relies on the fact that all non-zero elements of $\mbZ_p$ can generate the whole $\mbZ_p$.

	\begin{Claim}\label{claim:HC}
		If\, $\V'$ returns \yes to a proposed $\pi$, the vertices $(\pi^{-1}(0), \pi^{-1}(x), \pi^{-1}(2x), \dots, \pi^{-1}((p-1)x))$ in that order, as later outputted by $\mcB$, indeed form a Hamiltonian cycle of the graph $G$,
        regardless of whether $\pi$ is a correct bijection.
	\end{Claim}
	\begin{proof}
	When $\V'$ returns \yes, it comes to the \textbf{else} branches of both the outer and inner \textbf{if-then-else} statements. The inner statement implies that all edges $(\pi^{-1}(ix), \pi^{-1}((i+1)x))$ exist.
    The outer statement implies that $x \neq 0$, and therefore $\{0, x, 2x, ..., (p-1)x\} = \{0, 1, ..., p-1\}$ since each non-zero is a generator of $\mbZ_p$ for prime $p$. Combining the two gives a claimed Hamiltonian cycle.
	\end{proof}
	
By Claim \ref{claim:GroupIvio}, we know that before $\V'$ returns \yes, all answers of $\V'$ to $\mcA$'s trials are valid. And Claim \ref{claim:HC} guarantees that once $\V'$ returns \yes, $\mcB$ already finds a Hamiltonian cycle (and terminates the program). Note that when $\V'$ returns \yes, it may not be a correct response to $\mcA$, but now we do not care the correctness or even the completeness of the execution of $\mcA$ any more, because we have already used $\mcA$'s code to serve our purpose of finding a Hamiltonian cycle for $G$.
 	
The analysis so far shows that the algorithm $\mcB$ correctly outputs a Hamiltonian cycle as long as $\V'$ answers \yes. To finish the proof, we need to address the case that $\mcA$ halts before $\V'$ answers \yes. First, since the graph $G$ does have a Hamiltonian cycle by promise, the group $T$ as defined does exist (despite the fact that $\mcB$ could not really find it) and it is indeed isomorphic to $\mbZ_p$. Hence, a correct algorithm $\mcA$ cannot output \no.
 	
Now the only left case is that $\mcA$ may somehow infer, from the violations returned by $\V'$, a valid bijection $\sigma$ that maps $T$ to $\mbZ_p$ before $\V'$ returns \yes. In such a case, we can actually identify each $b_i = \sigma^{-1}(i\sigma(a_1))$: First, $b_0 = \pi^{-1}(0)$ since that is the only element which does not change by multiplying itself. Then, for all $i\geq 1$, we have
 	\begin{align}
 		\sigma(b_i) & = \sigma(\underbrace{b_1\circ ... \circ b_1}_{i \ b_1\text{'s}}) & (\text{def of $\circ$ in } T) \\
 		& = \underbrace{\sigma(b_1)+...+\sigma(b_1)}_{i \text{ times}} & (\sigma \text{ is a isomorphism}) \\
 		& = i\cdot\sigma(b_1) = i\cdot\sigma(a_1)
 	\end{align}
Thus, $\big(\sigma^{-1}(0), \sigma^{-1}(\sigma(a_1)), \sigma^{-1}(2\sigma(a_1)), \dots, \sigma^{-1}((p-1)\sigma(a_1))\big) \\ = (b_0, b_1, \ldots, b_{p-1})$, as outputted by algorithm $\mcB$, is a Hamiltonian cycle.
 	
Finally, the runtime of algorithm $\mcB$ is easily upper bounded by that of $\mcA$ times that of $\V'$. So it takes $\mcB$ at most $O(t(p)\cdot p)$ time to solve the $\PHC$ problem.
Now given a $\HC$ problem instance with $p$ nodes, we can run algorithm $\mcB$ on this instance and incorporate its polynomial time bound as a time limit. After the time limit, we check the output of the algorithm $\mcB$, and accept if it is indeed a valid Hamiltonian cycle. If algorithm $\mcB$ outputs a wrong Hamiltonian cycle or it fails to output one, we reject it. Thus in this way the $\HC$ problem can be solved in $O(t(p)\cdot p)$ time too. This completes the proof.
\end{proof}

A few remarks about the above proof are in order.
\begin{enumerate}

	\item The polynomial-time algorithm $\mcB$ cannot find a certificate of an $\np$ statement but can create a simulator for its own purpose. This phenomenon is reminiscent of the simulator paradigm in some cryptographic notions such as zero-knowledge proofs. However, differences are also clear, because simulators in zero-knowledge proofs are used to show that the interaction to the prover is in some sense ``useless'', while the simulator in our proof is actually useful for forcing $\mcA$ to leak information of a witness (in our case, a Hamiltonian cycle).
	
	
	\item One can generalize the theorem by allowing the algorithm $\mcA$ to also invoke a computation oracle $C$. Then our algorithm $\mcB$ can invoke the same oracle $C$ during simulating $\mcA$, and thus the conclusion becomes that ``if $\GpIp\uu$ can be solved in time $t(n)$ with computation oracle $C$ and verification oracle $\V$, then $\HC$ can be solved in time $O(t(n) n^2)$ with computation oracle $C$''. This implies that $\GpIp\uu$ is not likely to be solved in polynomial time by the help of a computation oracle weaker than $\np$.
\end{enumerate}

	



%% file: Graph-ISO.tex
\section{Graph Isomorphism}

In the graph isomorphism (\GraphI) search problem, we are given two undirected graphs
$G_1$ and $G_2$, and we are asked to find a bijection $\pi: V(G_1) \rightarrow V(G_2)$ \st $\forall i,j\in V(G_1)$,
\begin{equation}\label{eq:GraphIDef}
	(i,j)\in E(G_1) \Leftrightarrow (\pi(i),\pi(j))\in E(G_2),
\end{equation}
if such a permutation exists, and to output ``$G_1 \ncong G_2$'' otherwise. The graph isomorphism is a well-known \np problem whose complexity is still open.

In our unknown-input version, $\GraphI\uu$, we can consider both the case that the two graphs are unknown, and the case that exactly one graph, say $G_1$, is unknown. The latter case corresponds to the situations where we want to compare an unknown object (such as a new chemical compound) to a known one (such as a known chemical compound).

In either setting, we can propose a bijection $\pi: V(G_1) \rightarrow V(G_2)$. If it is indeed an isomorphism, the verification oracle returns \yes; otherwise, it returns a pair $i,j\in V(G_1)$ violating the above equivalence Eq.\eqref{eq:GraphIDef}.\footnote{Note that a return of $\V$ only implies that exactly one edge $(i,j)\in E(G_1)$ or $(\pi(i),\pi(j))\in E(G_2)$ exists, but does not tell which one. One may consider to define a stronger $\V$ revealing this further piece of information, but we will show that this distinction does not matter: Our algorithm works with the weaker $\V$, and our hardness result holds even for the stronger \V.}
Our upper bound result in this section applies to the model when both graphs are unknown, and our lower bound result applies to the model when one graph is known. Therefore, both of our upper and lower bound results are at the stronger sense.


A natural try for an algorithm is, as in the algorithms for $\SM$ and $\Sat$, to keep proposing bijections $\pi$ that are consistent with the current knowledge of the edge information of the two graphs. More precisely, each returned violation $(i,j)$ implies that a homomorphism, if one exists, should not \emph{simultaneously} map $i$ to $\pi(i)$ and map $j$ to $\pi(j)$ (or $i$ to $\pi(j)$ and $j$ to $\pi(i)$). Actually we can do better: Consider a bipartite graph $H$ with $\binom{n}{2}$ nodes at each side, where the left and right hand side nodes are indexed by all pairs of different vertices $(i,j)$ in $G_1$ and $G_2$, respectively. Starting from empty, $H$ is updated as follows. Each time we propose a bijection $\pi$ and $\V$ returns a pair $(i,j)$, we add an edge $((i,j),(\pi(i),\pi(j)))$ in $H$. Then an efficient algorithm tries to propose a bijection $\pi$ \st $(i,j)$ and $(\pi(i),\pi(j)))$ are not in the same  connected component in the current $H$, so that any newly returned $(i,j)$ gives a new piece of information. More precisely, each added edge in $H$ decreases the number of connected components by 1. Since initially all $2\binom{n}{2}$ nodes in $H$ are isolated and finally the nodes form at least one components, the algorithm stops after at most $2\binom{n}{2}-1$ trials.

The above analysis gives a trial-efficient algorithm to solve $\GraphI\uu$. When it comes to the time complexity, we need to address the question of how to find a $\pi$ to avoid a collection of forbidden pairs. This can surely be done if we are given an $\np$ oracle since \emph{checking} a bijection avoiding a list of forbidden pairs is easy. This leads to the following proposition.

\begin{Prop}
	$\GraphI\uu$ can be solved using $\Sat$ as the computation oracle in polynomial time and by $O(n^2)$ trials.
\end{Prop}



Of course, as in $\Sat_u\in \p^{\scriptsize \V,\Sat}$, one naturally desires an algorithm using only $\GraphI$ as the computation oracle for $\GraphI\uu$. It looks quite achievable: After all, graphs is by nature a collection of binary relations, and the well-developed graph theory is a large source of tools. Indeed, it is not hard to show by a simple probabilistic argument that one can propose a $\pi$ to avoid $\Theta(n)$ existing forbidden pairs, so it is ``merely" the matter of whether the process can continue to handle more forbidden pairs. However, these intuitions turn out to be wrong, as refuted by the following theorem about the necessity of the $\Sat$ computation oracle.

\begin{Thm}\label{thm:GraphI}
If $\GraphI\uu$ can be solved in time $t(n)$ with computation oracle $A$ and verification oracle $\V$, then $\Clique$ can be solved in time $O(t(n) n^2)$ with oracle $A$.
\end{Thm}
\begin{proof}
We prove the claim for the model when one graph is known.
We assume that there is an algorithm $\mathcal{A}$ that solves $\GraphI\uu$ in polynomial time.
We will use it to design another algorithm $\mcB$ to solve the $\Clique$ problem --- given a graph $G$ and a number $k$, find a clique of size $k$ in $G$, or claim that it does not exist.

For the given $G$, we construct a $\GraphI\uu$ instance as follows: Let the known graph $G_2=G$ and the unknown graph be $G_1$ where $V(G_1)=\{1,\ldots,n\}$ and $E(G_1)=\{(i,j)~|~1\le i<j\le k\}$, \ie, $G_1$ is composed of a $k$-clique with extra $n-k$ isolated nodes. 
We now apply algorithm $\mathcal{A}$ on the instance $(G_1,G_2)$ with the following specific verification oracle $\mathcal{O}$ upon a query: Given a bijection $\pi$ from $V(G_1)$ to $V(G_2)$, if it is indeed an isomorphism, then return \yes. Otherwise $\mcO$ returns a pair $(i,j)$ that minimizes $\max\{i,j\}$ where the minimization is over all pairs $(i,j)$ \st $(i,j) \in E(G_1)$, and $(\pi(i),\pi(j)) \notin E(G_2)$. If no such pair exists, it returns an arbitrary violated pair (for such a pair $(i,j)$, it must be $\max\{i,j\}>k$). Note that $\mathcal{A}$ runs in polynomial time on the instance and $\mathcal{O}$ is efficiently implementable.

We next construct an algorithm $\mathcal{B}$ for the $k$-clique problem.

\begin{algorithm}
  \caption{Algorithm $\mathcal{B}$ for $k$-clique}
    {\bf Input}: Graph $G$
  \begin{algorithmic}[1]
    \STATE $G_1 \leftarrow ([n],\{(i,j)~|~1\le i<j\le k\})$, $G_2\leftarrow G$.
    \STATE Run algorithm $\mathcal{A}$ on $(G_1,G_2)$ on oracle $\mathcal{O}$.
        \IF {($\mathcal{A}$ outputs a bijection $\pi$) or ($\mathcal{O}$ returns a pair $(i,j)$ with $\max\{i,j\} > k$ on $\mcA$'s query $\pi$)}
        \STATE Return ``$(\pi(1), \ldots, \pi(k))$ is a clique in $G$''.
        \ELSE
        \STATE Return ``no $k$-clique''.
        \ENDIF
  \end{algorithmic}
\end{algorithm}

If $G$ is precisely a $k$-clique plus $n-k$ isolated nodes, then $G_1$ is isomorphic to $G=G_2$, and $\mathcal{A}$ finally returns a correct bijection; thus $\mathcal{B}$ also finds the clique.
If $G$ contains a $k$-clique as well as some other edges, then $G_1$ is not
isomorphic to $G$ 
and $\mathcal{A}$ finally outputs \no. This looks undesirable since $\mcB$ may also output ``no $k$-clique'' in the \textbf{else} branch. However, we claim that the verification oracle $\mathcal{O}$ always returns a pair $i$ and $j$ with $\max\{i,j\} > k$ before $\mathcal{A}$ can conclude with an answer \no. Indeed, to conclude that $G_1$ is not isomorphic to $G$, $\mathcal{A}$ has to detect at least one pair $(i,j)$ with $\max\{i,j\}>k$ and $(i,j)$ is not an edge in $G_1$. (Otherwise, $\mcA$ only see edges within the first $k$ nodes in $G_1$ and it is still possible in $\mathcal{A}$'s point of view that $G_1$ is isomorphic to $G$. Thus $\mathcal{A}$ cannot make any decision yet.) But the only way that $\mathcal{A}$ detects such a pair $(i,j)$ is when $\mathcal{O}$ returns a pair $(i,j)$ where $(\pi(i),\pi(j)) \in E(G)$ but $(i,j) \notin E(G_1)$. By the design of our oracle, this can happen only if $\max\{i,j\} > k$, as claimed.

Once $\mcO$ returns a pair $(i,j)$ with $\max\{i,j\} > k$, all the edges $(\pi(i),\pi(j))$ in $G$ exist; thus, $\mcB$ already finds a $k$-clique and outputs it.

On the other hand, if $G$ does not contain a $k$-clique, then $\mathcal{A}$ never outputs a bijection. Further, for any
$\pi$, there is at least one pair $(i,j)$, $1\le i<j\leq k$, such that $(i,j)\in E(G_1)$ and
$(\pi(i),\pi(j)) \notin E(G)$, and the oracle $\mathcal{O}$ as defined
always returns one of such pairs. Hence, the algorithm $\mathcal{B}$ never outputs a $k$-clique
before $\mathcal{A}$ finishes, at which time $\mathcal{B}$ gets to the last line and outputs ``no $k$-clique''.

Finally, for the time cost, each execution of $\mcO$ takes time $k^2 \leq n^2$ and the number of queries to $\mcO$ in $\mcA$ is at most its time complexity $t(n)$, so the total time on $\mcO$ is $O(t(n) n^2)$. Other time cost mainly includes the non-query part of $\mcA$, which is at most $t(n)$, so claimed time bound holds.
\end{proof}

We have the following immediate corollary.

\begin{Cor}
For any given computation oracle $L$ (which is a class of languages), $\GraphI\uu \in \p^{\scriptsize {\V,L}}$ if and only if $\np \subseteq \p^{L}$.
\end{Cor}
\begin{proof}
  First, if $\np \subseteq \p^{L}$, which means $\p^{{\scriptsize \np}} \subseteq \p^{L}$, would imply $\GraphI\uu \in \p^{\scriptsize {\V,\np}} \subseteq \p^{\scriptsize {\V,L}}$.
  Second, Theorem \ref{thm:GraphI} will directly give us that if $\GraphI\uu \in \p^{\scriptsize {\V,L}}$, then $\Clique \in \p^{L}$, which means $\np \subseteq \p^{L}$.
\end{proof}

It was shown in~\cite{BHZ87} that if \GraphI is \np-complete, then the polynomial hierarchy (\ph) collapses to the second level. The proof can be easily adapted to show a slightly stronger result that \ph collapses to the second level even if $\GraphI$ is \np-complete under Turing reduction\footnote{Eric Allender later pointed out that this stronger result, as we guessed, was indeed known, e.g., in \cite{Sch87}.}. This gives the following corollary. We include an elementary proof for completeness.

\begin{Cor}
If $\GraphI\uu \in \p^{{\scriptsize \V,\GraphI}}$, then the polynomial hierarchy (\ph) collapses to the second level.
\end{Cor}

\begin{proof}
It is sufficient to prove that $\Sigma_2 = \Pi_2$. We will show the inclusion $\Sigma_2 \subseteq \Pi_2$, and the other direction is similar. For any formula $\phi(x,y)$, where $x,y\in \Bn$, we want to construct another formula $\phi'(r_1,\ldots, r_k;x,a_1,\ldots, a_k)$, where $k$ and the lengths of all $r_i$'s and $a_i$'s are $poly(n)$, s.t.
\begin{eqnarray}
\exists x \forall y\, \, \phi(x,y) = 1 \Leftrightarrow \label{eq:GI-NPC} \forall (r_1\ldots r_k) \exists (x,a_1\ldots a_k) \, \, \phi'(r_1,\ldots, r_k,x,a_1,\ldots, a_k) = 1. 
\end{eqnarray}
Next is the construction.

By Theorem \ref{thm:GraphI}, if $\GraphI\uu \in \p^{{\scriptsize \V,\GraphI}}$, then $\np\subseteq \p^{{\scriptsize \GraphI}}$, \ie, any $\np$ problem can be solved by a polynomial-time algorithm calling the $\GraphI$ oracle at most $k = poly(n)$ times. By flipping the answer, the algorithms can also solve $\co\np$ problems. Note that $\forall y \ \phi(x,y) = 1$ is a $\co\np$ statement, so it can be solved by a polynomial-time algorithm with the $\GraphI$ oracle.

It is well-known that there is an $\am$ protocol for $\mbox{\sf {GraphNonIso}}$ with perfect completeness and soundness error less than $2^{-m}$, for any $m$ polynomial in the length of the input. Also note that there is a trivial $\np$ proof for $\GraphI$, which is a special case of an $\am$ protocol with perfect completeness and perfect soundness. So we can design a protocol of $2k$ rounds to solve a $\co\np$-complete problem. Basically, the verifier simulates the algorithm mentioned in the above paragraph for the $\co\np$ problem. When it comes to the $i$-th query to the $\GraphI$ oracle, the verifier sends $r_i$ as if it is the $\am$ protocol for $\mbox{\sf {GraphNonIso}}$. Since it is a public-coin protocol, the verifier's code is deterministic except for the public random coins sent to the prover. So each time the prover knows the pair of graphs currently in the verifier's mind (as the input for the $\GraphI$ oracle). So the prover is supposed to solve the graph isomorphism problem and to return the one-bit answer, followed by a proof of that answer.

Now we will define a polynomial-time verification process $V'$ on input $(r_1,\ldots, r_k,x,a_1,\ldots, a_k)$, and then take $\phi'$ to be the formula induced by $V'$ as in the standard Cook-Levin reduction, and show Eq.\eqref{eq:GI-NPC}. Let $V$ denote the predicate which the assumed algorithm for the \co\np problem uses, after all queries to the \GraphI oracle, to decide acceptance/rejection. Now let $V'$ on $(r_1,\ldots, r_k,x,a_1,\ldots, a_k)$ be the following: Check each $a_i$ is a valid answer respect to $r_i$, and if all pass, output $V(r_1,\ldots, r_k,x,a_1,\ldots, a_k)$.

Since no matter whether the answer is 0 or 1, the protocol always has perfect completeness. Therefore, for the Yes instances of the original $\Sigma_2$ language, we have
\begin{align}
	& \exists x \forall y \, \, \phi(x,y) = 1 \\
	\Rightarrow & \exists x \forall r_1 \exists a_1 \ldots \forall r_k \exists a_k \quad V'(x,\phi,r_1, \ldots, r_k,a_1,\ldots,a_k) = 1 & (\text{due to perfect completeness}) \\
	\Rightarrow & \exists x \forall r_1 \ldots \forall r_k \exists a_1 \ldots \exists a_k \quad V'(x,\phi,r_1, \ldots, r_k,a_1,\ldots,a_k) = 1 & (\text{use the honest prover}) \\
	\Rightarrow & \forall r_1 \ldots \forall r_k \exists x \exists a_1 \ldots \exists a_k \quad V'(x,\phi,r_1, \ldots, r_k,a_1,\ldots,a_k) = 1 & (\text{use the fixed $x$})
\end{align}

On the other hand, for the No instances of the $\Sigma_2$ language, we have
\begin{align}
	& \forall x \exists y \, \, \phi(x,y) = 0 \\
	\nonumber \Rightarrow & \forall x, \text{for }(1-2^{-m})\text{-fraction of } r_1, \forall a_1,\ldots, \text{for }(1-2^{-m})\text{-fraction of } r_k, \forall a_k, & \\
	& \quad V'(x,\phi,r_1, \ldots, r_k,a_1,\ldots, a_k) = 0 \hspace{5.9em}
	(\text{small soundness error for each round}) \\
	\nonumber \Rightarrow & \text{for }(1-k2^{n-m})\text{-fraction of } (r_1, \ldots, r_k), \forall x \forall a_1 \ldots \forall a_k \\
	& \quad V'(x,\phi,r_1,\ldots,r_k,a_1,\ldots,a_k) = 0  \hspace{16.5em}
	(\text{union bound}) \\
	\Rightarrow & \exists r_1 \ldots\exists r_k \forall x \forall a_1 \ldots \forall a_k \quad V'(x,\phi,r_1,\ldots,r_k,a_1,\ldots,a_k) = 0 \quad (\text{whenever}\ m > n+\log_2 k)
\end{align}
So if we pick $m = n + \lceil \log_2 k\rceil  + 1$, then Eq.\eqref{eq:GI-NPC} holds, as desired.
\end{proof}

%% file: nash.tex
\section{Nash Equilibrium}\label{section-nash}



In a normal-form game, there are $n$ players. Each player $i$ has a strategy space $S_i$ and a payoff function $u_i:
S_1\times\cdots\times S_n \mapsto \mathbb{Q}$, which gives the utility that $i$ obtains for every strategy profile $(s_1, \ldots, s_n)\in S_1\times\cdots\times S_n$
. 
A joint probability distribution $(p_1,\ldots,p_n)$ on $S_1\times\cdots\times S_n$ is called a {\em (mixed) Nash equilibrium} if
for any player $i$ and any probability distribution $p'_i$ on $S_i$, we have
\begin{eqnarray}
\sum_{(s_1, \ldots, s_n)} \prod_{j}p_j(s_j)\cdot u_i(s_1, \ldots, s_n) &\ge& \sum_{(s_1, \ldots, s_n)} p'_i(s_i)\cdot \prod_{j\neq i}p_{j}(s_{j})\cdot u_i(s_1, \ldots, s_n).
\end{eqnarray}
Note that the number of constraints given by the foregoing inequality is unbounded.
It is well-known that a two-player game admits a mixed Nash equilibrium with
polynomial size rationals, whereas games with three or more players
may only have equilibria in irrational numbers~\cite{CD86}. The \nash problem is to find a Nash equilibrium in a normal-form game.

In the unknown-input version of a given game, denoted by $\nash\uu$, the payoff functions
$u_i(\cdot)$ are unknown. We can query a mixed strategy $(p_1, \ldots, p_n)$ each time. If it is not an equilibrium, then the
oracle will return a player $i$ and one of his better responses
$p'_i$ where the foregoing inequality fails to hold (note that $i$ and $p'_i$ are precisely the index of a violated constraint).\footnote{Note that the full set of strategies $S_i$ may also be unknown; that is,
in the process of trials, we query a probability distribution over those strategies that we have already observed. A deviation from a player
can be either from the known strategies or ``new'' unknown strategies.}
Our trial and error model considers how fast a Nash equilibrium can be found from the viewpoint of a centralized authority, which is quite different from the learning models investigated in~\cite{FL98,You09} whose focuses are on the strategic dynamics formed by the behavior of individual players.
We have the following result.

\begin{Thm}\label{theorem-nash}
There is a polynomial-time algorithm solving $\nash\uu$ for any two-player game, given a computation oracle solving \nash.  
\end{Thm}

\noindent
{\em Idea of the proof.} The proof is built on the existence of a Nash equilibrium in any game~\cite{Nash51}.
  Assume that each player has $m$ strategies. There are a total of $2m^2$
  values in the two payoff matrices. Note that the Nash equilibrium solution space may not be convex; thus, we cannot employ the ellipsoid method to search for an equilibrium in the solution space.
  One observation is that the $2m^2$ values in the matrices
  correspond to a point $U$ in the space $\mbR^{2m^2}$, which can also
  be seen as a degenerate polyhedron in $\mbR^{2m^2}$.
  For any given point $X \in \mbR^{2m^2}$, we can consider it as two
  payoff matrices of some game. If we compute a Nash equilibrium with
  respect to this game using the computation oracle and query it to
  the verification oracle, then the returned information (if it is not a
  $\yes$) actually gives us a hyperplane that separates $X$ from the
  true point $U$. It is now tempting to claim that the problem is solved by the ellipsoid method.
  However, there is a remaining issue: in our problem, the solution polyhedron degenerates to a point and has
  volume 0. The standard approach in the ellipsoid method for handling such degenerated cases is to add perturbations to the
  constraints to introduce a positive volume of the feasible solution
  polyhedron. However, this approach is not applicable in our context,
  as we do not know the constraints explicitly. Luckily, we are able
  to employ a much more involved machinery developed by Gr{\"o}tschel,
  Lov{\'a}sz, and Schrijver~\cite{GLS84,GLS88}, solving the strong nonemptiness problem for well-described polyhedra given by a strong separation oracle, to overcome this issue
  and thus solve the problem. \hfill $\square$

\begin{proof}[Formal Proof of Theorem~\ref{theorem-nash}]
  Assume without loss of generality that each player has $m$
  strategies. Notice that there are totally $2m^2$ values in the
  payoff matrices given by $u_1$ and $u_2$; these values correspond to a point $U$ in the space
  $\mbR^{2m^2}$. If we are able to construct a separation oracle for $U$
  (which can also be seen as a degenerated polyhedron), then we
  can apply the ellipsoid method to find this point and thus solve the problem. Therefore the remaining
  problem is how to construct such a separation oracle in polynomial time.

  For any given point $X \in \mbR^{2m^2}$, consider it as two utility
  matrices $x_1, x_2$ of another game. We can first compute a Nash equilibrium
  $(p_1, p_2)$ with respect to the two new utility matrices. Next we query this mixed
  strategy $(p_1, p_2)$ to the verification oracle. If $(p_1, p_2)$ is already a Nash equilibrium to the unknown utility matrices $u_1$ and $u_2$, then the verification oracle tells us so, and we have thus solved the problem. Now we assume that $(p_1, p_2)$ is not a Nash equilibrium to the utility matrices
  $u_1$ and $u_2$. In this case we know that $X \neq U$. Suppose
  without loss of generality that the verification oracle returns that the first player
  has a better response $p_1'$. This means that
  $$\sum_{(s_1, s_2)}p_1(s_1)\cdot p_2(s_2)\cdot u_1(s_1, s_2) < \sum_{(s_1, s_2)}p_1'(s_1)\cdot p_2(s_2)\cdot u_1(s_1, s_2).$$
  Also notice that $(p_1, p_2)$ is a Nash equilibrium to utility
  matrices $x_1$ and $x_2$, thus we have
  $$\sum_{(s_1, s_2)}p_1(s_1)\cdot p_2(s_2)\cdot x_1(s_1, s_2) \geq \sum_{(s_1, s_2)}p_1'(s_1)\cdot p_2(s_2)\cdot x_1(s_1, s_2).$$

  This means that the vector $(p_1-p_1')\otimes p_2$, whose
  $(s_1,s_2)$-entry is $(p_1(s_1)-p_1'(s_1))p_2(s_2)$, can serve as
  the returned vector for the strong separation oracle (separating
  point $U$ from point $X$). Thus, we can use an approach developed by
  Gr{\"o}tschel, Lov{\'a}sz, and Schrijver~\cite{GLS84,GLS88} that
  based on the ellipsoid method to solve this problem. Formally, the
  approach can be used to provide a polynomial-time algorithm for the
  strong nonemptiness problem for well-described polyhedra $K$ given
  by a strong separation oracle (\cite{GLS88}, Theorem 6.4.1). Here a
  strong nonemptiness problem is to decide whether $K$ is empty, and
  if not, finding a point in $K$. Our (degenerated) polyhedron $U$ is
  a single point and is thus well-described. A strong separation
  oracle is one that on a given point $y\in \mbR^n$ and a polyhedron
  $K$, finds a vector $c$ such that $c^T y > \max\{c^T x: x\in K\}$.
  As argued above, we can construct such a strong separation oracle
  for point $U$ easily in polynomial time (thanks to the computation
  oracle for solving \nash); thus, we can identify the exact value of
  these two utility matrices $u_1$ and $u_2$ (or find a Nash equilibrium in
  some middle step). Once we find $u_1$ and $u_2$, we call the computation
  oracle one more time to compute a Nash equilibrium.
\end{proof}

We would like to comment that here the polyhedron $U$ is
degenerated and has volume 0. In the known-input case, one can use the
standard perturbation approach in the ellipsoid method to introduce a
positive volume to $U$ for handling this issue. However, the approach
is not applicable in our context, as we do not know the constraints
explicitly. Therefore, we have to use a much more involved machinery
in \cite{GLS84,GLS88}, as applied in the above proof, to find
halfspaces that contain $U$, and do a sequence of dimension
reductions. Fortunately and very interestingly, what is provided by
the verification oracle just fits what the method requires.

For completeness, we briefly describe the idea of this general method.
Let $F$ denote the feasible solution polyhedra.
Given the separation oracle provided by the verification oracle, we use the ellipsoid method to ask queries iteratively: Initially we pick an ellipsoid
that covers the entire domain $[0,N]^n$ and query the center of the
ellipsoid. If a constraint is violated from the verification oracle, we
establish a separation oracle and compute the next ellipsoid. (Note
that a remarkable property of the ellipsoid method is that we do not
even need to know the explicit expression of the violated
constraint, and a separation oracle suffices for the algorithm to
continue.) The process continues until either we find a point which
is in $F$ or the volume of the ellipsoid is sufficiently small such
that $F$ must have volume 0.

Note that if the system of inequalities has a solution, the numerator and denominator of all its components are
bounded by $(nN)^n$. Thus, there is a lower bound on the volume of
$F$ if it is positive. After the volume of the ellipsoid gets
smaller than that lower bound, which can be done in polynomial to
$n$ and $\log N$ number of steps, we can conclude that $F$ is not
full-dimensional, i.e., all points in $F$ are lie on a hyperplane
$H$. If we can find the hyperplane $H$ in polynomial time, we know
an extra (explicit) constraint and can reduce our problem to an
$(n-1)$-dimensional case; then we can use the same method to solve it
recursively.

Thus the remaining problem is how to identify the hyperplane $H$ in
polynomial time. This problem is solved by Gr{\"o}tschel,
Lov{\'a}sz, and Schrijver in \cite{GLS84}. The general idea is that,
having the ellipsoid with small enough volume that contains the
solution polyhedra $F$, the ellipsoid must be very ``thin'' in the
direction perpendicular to $H$. That is to say, if we take the
shortest axis of this ellipsoid, let $u$ be the unit vector which is
parallel to this axis and $v$ be the center of the ellipsoid, the
hyperplane $u^Tx = u^Tv$ must be very close to our target hyperplane
$H$. Then next we use the simultaneous diophantine approximation
algorithm, which is a technique to round real numbers by rational
numbers with relatively small sized denominator, to round the
coefficients of $u^Tx = u^Tv$; this will finally give us the
hyperplane $H$.


\medskip
Although the aforementioned claim applies only to two-player games, our approach can also be generalized to $n$ players
to obtain an $\epsilon$-approximate mixed Nash equilibrium for any constant $\epsilon>0$.
(Note that we cannot hope to compute an exact Nash equilibrium when $n\ge 3$, as such a solution may consist of irrational numbers.)
However, there is one potential issue: the input size of a game can be as large as $O(m^n)$, where $m$ is the number of strategies of each player.
Thus, for general games, our algorithm may require running time polynomial to $O(m^n)$ (which is still polynomial in the input size).
However, for some multi-player games that admit concise representations, e.g., graphical games~\cite{KLS01} on constant degree graphs, we can find an $\epsilon$-approximate Nash equilibrium in time polynomial to $m$ and $n$.

Our result implies that even if players are not completely aware of the rules of a game, we can still find a Nash equilibrium efficiently. Further, even if a Nash equilibrium has been achieved, the game itself can still remain a mystery (because beyond this point, the verification oracle cannot return any further information). The Internet provides one such example, as Scott Shenker remarked:
``{\it The Internet is an equilibrium, we just have to identify the game}''~\cite{Papa01}.

In addition, we note that our approach can also be adopted to solve some other similar problems such as correlated equilibrium in games and competitive equilibrium in matching markets with prices (the details are quite similar to $\nash\uu$ and thus are omitted here).

%% file: core.tex
\section{Core of Cooperative Games}\label{sec:core}

In a cost-sharing game, we are given a set $A$ of $n$ agents, and a
cost function $c: 2^A \mapsto \mbR^+\cup \{0\}$. Basically, the cost function gives the cost
for any subset $S \subseteq A$ in order to let every agent in $S$ be
served. For example, $c(S)$ can be the cost to build a network that connects everyone in set $S$ to the Internet. We assume that $c(\emptyset)=0$ and the function is monotone, i.e., $c(S)\le c(T)$ if $S\subseteq T$. We say a vector
$\mathbf \valpha = (\alpha_i)_{i\in A} \in \mbR^n$ is in the
\textit{core} of the game if it satisfies the following two
conditions:
\begin{itemize}
\item $\sum_{i \in A}\alpha_i = c(A)$.
\item $\sum_{i \in S}\alpha_i \leq c(S)$ for every $S \subset A$.
\end{itemize}
Core is a central notion in cooperative game theory. The classic
Bondareva-Shapley Theorem~\cite{bondareva63,shapley67} says that the
core of a cooperative game is non-empty if and only if the cost
function is fractionally subadditive. (A function is fractionally subadditive if there is a set of
linear functions $f_1,\ldots,f_m$ such that $v(S)=\max\big\{f_1(S),f_2(S),\ldots,f_m(S)\big\}$ for any $S\subseteq A$.) Fractionally subadditive functions form a pretty general class which
includes, e.g., additive functions, gross substitutes functions, and submodular functions as
special cases~\cite{LLN06}.

In the corresponding unknown-input problem, denoted by $\Core\uu$, the cost function $c(\cdot)$ is unknown; the information that we have is the number of agents $n$ and an integer
$N$ which bounds the encoding length of every $c(S)$. (We assume
that each $c(S)$ is given by two rational numbers presenting its
numerator and denominator, respectively, whose values are therefore
bounded by $2^N$.) We can propose a vector $\valpha
\in [0,N]^n$. If it is in the core, the verification oracle returns a \yes answer;
otherwise, it returns the index of a violated constraint. Here the set of constraints contains precisely the linear constraints used to define the core, except that we replace the
inequality $\sum_{i \in A}\alpha_i = c(A)$ by two inequalities:
$\sum_{i \in A}\alpha_i \ge c(A)$ and $\sum_{i \in A}\alpha_i \le
c(A)$.\footnote{This is necessary to admit a polynomial time
algorithm. For instance, when there is only one agent, the problem
degenerates to find a given unknown rational number using queries.
If the query is of the form ``Is $x=y$?'', it will take an exponential
number of queries in the worst case; but if the query is of the form
``Is $x\le y$?'', polynomial time algorithms are
known~\cite{Papa79,KM03}.}

\begin{Thm}\label{thm-cost-sharing}
There is an algorithm solving the $\Core\uu$ problem in time polynomial in the input size.
\end{Thm}

Note that when the cost function is additive, i.e., $\sum_{i\in
A}c(i)=c(A)$, there is a unique solution to the core, i.e.,
$\alpha_i=c(i)$. An implication of our result is that we are able to
identify these number $c(i)$'s precisely given the verification
oracle. Further, together with Bondareva-Shapley
Theorem~\cite{bondareva63,shapley67}, our result immediately implies
that we are able to determine if a given function is in the class of
fractionally subadditive in polynomial trials.

\begin{proof}[of Theorem~\ref{thm-cost-sharing}]
According to the definition, the core of a cost-sharing game is just the set $F$ of solutions
described by the following system of linear inequalities:
\begin{eqnarray*}
   \sum_{i \in A}-x_i & \le & -c(A). \\
   \sum_{i \in A}x_i & \le & c(A). \\
   \sum_{i \in S}x_i & \leq & c(S), \qquad \forall S \subset A.
\end{eqnarray*}
This is a system of linear inequalities with $n$ variables and $2^n+2$ constraints.

Each time when we propose a vector $\valpha$, either we know that
$\valpha \in F$, which implies that $\valpha$ is in the core and we
are done), or we get a subset $S$ where $\sum_{i \in S}\alpha_i >
c(S)$. (The case that the first inequality is violated can be handled
similarly.) Note that the value of $c(S)$ is still unknown to us, but
we are able to get a strong separation oracle for $(F,\valpha)$.
Because our polyhedron $F$ has short representation and is thus
well-described. Hence again we can apply the ellipsoid algorithm 
either to
find a point $x \in F$ or to conclude that $F$ is empty, solving the
problem $\Core\uu$.
\end{proof}



%% file: subset.tex
\section{Subset Sum}

Given a set $S=\{a_1,\ldots,a_n\}$ of $n$ integers, the subset sum problem is to find a partition of $S$ into two subsets $S_1$ and $S_2$ such that
$\sum_{a\in S_1}a = \sum_{b\in S_2}b$, or report that such a partition does not exist.

In the unknown input version of subset sum problem, denoted by $\ssum\uu$, the values of these $n$ integers are unknown to us.
Each time we can propose a partition $S_1$ and $S_2$. If it is indeed a solution to the subset sum problem,
the verification oracle will return $\yes$. Otherwise, the oracle will return which subset has a larger total value. That is,
the violation is one of the following two constraints.
\begin{itemize}
\item $\sum_{a\in S_1}a \ge \sum_{b\in S_2}b$
\item $\sum_{a\in S_1}a \le \sum_{b\in S_2}b$
\end{itemize}

\begin{Thm}
The $\ssum\uu$ problem has an exponential lower bound on trial complexity.
\end{Thm}

\begin{proof}
Consider the following instance of \ssum: $S=\{a_1,a_2,\ldots,a_{2n}\}$ of
  $2n$ integers, which can be divided into three categories:
  \begin{itemize}
  \item $a_1=M+n+2$.
  \item $a_2=\cdots = a_n = M+2$.
  \item $a_{n+1}=\cdots=a_{2n}=M+3$.
  \end{itemize}
  Here $M$ is a sufficiently large integer, such that in any partition, if the
  two sets have different sizes (number of integers), the larger set will always have a
  larger sum. Thus, it is easy to see that the only valid partition is given by $S_1 =
  \{a_1,a_2,\ldots,a_n\}$ and $S_2 = \{a_{n+1},\ldots,a_{2n}\}$.


Given an algorithm for $\ssum\uu$, since the instance $S$ has a valid and unique partition solution, it should be able to
  find out the subset $\{a_{n+1},\ldots,a_{2n}\}$ precisely, i.e., the exact partition $(S_1,S_2)$. For any query
  $(T_1, T_2)$, if $T_1$ and $T_2$ have different sizes, the
  oracle will always return the set with bigger size (which has a larger sum), and we cannot derive any information from this query. If
  $T_1$ and $T_2$ have the same size but are not exactly $S_1$ and
  $S_2$, it is always the case that the set containing $a_1$ has a larger sum. Thus, the only information we can derive
  is a subset of candidates of $a_1$. Therefore, in order to find the
  subset $\{a_{n+1},\ldots,a_{2n}\}$, even if we already know the membership of $a_1$, in the worse case one needs as much as
  ${2n-1 \choose{n}}$ queries; this gives the desired exponential lower bound.
\end{proof}